\documentclass[showpacs, amsmath, amssymb, amsthm, aps, jmp, thmsb]{revtex4}

\usepackage{color}
\usepackage{float}
\usepackage{amsmath}
\usepackage{amssymb}
\usepackage{amsfonts}
\usepackage{amsbsy}
\usepackage{latexsym}
\usepackage{mathrsfs}
\usepackage{bbm}
\usepackage{bm}
\usepackage[dvipdfm]{graphicx}

\newtheorem{definition}{Definition}
\newtheorem{proposition}{Proposition}
\newtheorem{theorem}{Theorem}
\newtheorem{lemma}{Lemma}
\newtheorem{corollary}{Corollary}

\newenvironment{proof}[1][Proof]{\textbf{#1.} }{\ \rule{0.5em}{0.5em}}

\newcommand{\be}{\begin{eqnarray}}
\newcommand{\ee}{\end{eqnarray}}
\def\({\left(}
\def\){\right)}
\def\[{\left[}
\def\]{\right]}
\def\C{\mathbb{C}}
\def\R{\mathbb{R}}
\def\vect#1{\mbox{\boldmath $#1$}}

\newcommand{\e}{\mathrm{e}}
\newcommand{\im}{\mathrm{i}}

\newcommand{\dd}{\mathrm{d}}

\newcommand{\braket}[1]{\left\langle #1 \right\rangle}

\newcommand{\bra}[1]{\langle #1 |}
\newcommand{\ket}[1]{| #1 \rangle}

\newcommand{\sla}[1]{\rlap{\kern .15em /}#1}

\newcommand{\ot}{\otimes}
\newcommand{\Bot}{\bigotimes}

\def\ann#1#2#3{{\it Ann.\ Phys.}\ {\bf {#1}} ({#2}) #3}

\def\lmp#1#2#3{{\it Lett.\ Math.\ Phys.} {\bf {#1}} ({#2}) #3}

\def\njp#1#2#3{{\it New\ J.\ Phys.}\ {\bf {#1}} ({#2}) #3}

\def\pr#1#2#3{{\it Phys.\ Rev.}\ {\bf {#1}} ({#2}) #3}
\def\prA#1#2#3{{\it Phys.\ Rev.}\ {\bf A{#1}} ({#2}) #3}

\def\prl#1#2#3{{\it Phys.\ Rev.\ Lett.}\ {\bf #1} ({#2}) #3}

\begin{document}
\title{Separability of $N$-particle Fermionic States for Arbitrary Partitions}

\author{
Tsubasa Ichikawa$^1$, Toshihiko Sasaki$^2$ and Izumi Tsutsui$^2$
}

\affiliation{
$^1$Research Center for Quantum Computing, Interdisciplinary Graduate School of Science and Engineering, 
 Kinki University, 3-4-1 Kowakae, 
Higashi-Osaka, Osaka 577-8502, Japan\\
$^2$Theory Center, 
Institute of Particle and Nuclear Studies, High Energy Accelerator Research Organization (KEK), 1-1 Oho, Tsukuba, Ibaraki 305-0801, Japan
}

\begin{abstract}
We present a criterion of separability for arbitrary $s$ partitions of $N$-particle fermionic  pure states.  We
show that, despite the superficial non-factorizability due to the antisymmetry required by the indistinguishability of the particles, 
the states which meet our criterion have factorizable correlations for a class of observables which are specified consistently with the states.
The separable states and the associated class of observables share an orthogonal structure, whose non-uniqueness is found to be intrinsic to the 
multi-partite separability and leads to the non-transitivity in the factorizability in general.  Our result generalizes 
the previous result obtained by Ghirardi {\it et.\,al.} [J.~Stat.~Phys.~{\bf 108} (2002) 49] for the $s=2$ and $s=N$ case.
\end{abstract}

\maketitle
\section{Introduction}

One of the distinguished features of quantum theory is the existence of indistinguishable particles.  
The indistinguishability of particles, which leads to 
the symmetry or antisymmetry of states depending on the particle statistics,  
demystifies miscellaneous physical phenomena and  properties,  
including black-body radiation, the electronic shell structure of atoms, Bose-Einstein condensation, superconductivity, 
and Hanbury Brown and Twiss effect.
Meanwhile, the indistinguishability mystifies us as well, especially in its perplexing 
consequence in quantum correlation or entanglement.   For example, the cluster decomposition property \cite{WC63, Peres93, Weinberg95}
of quantum field theory ensures that correlation functions of local observables are factorizable for indistinguishable particles, 
despite that the states of the particles are necessarily intertwined with each other due to the (anti-)symmetry prescribed.  
The puzzling aspect of indistinguishability of particles in quantum theory
derives from the fact that the entire Hilbert space is not just the tensor product of the constituent Hilbert spaces of individual particles; rather it is the (anti-)symmetrized subspace of it.   Consequently, the conventional criterion of  entanglement for distinguishable particles is no longer valid for indistinguishable particles.  

A definite criterion of entanglement for indistinguishable particles is indispensable also for quantum information science, where characterization of entangled  multi-partite states forms a key element to achieve quantum computation and communication effectively, or any other operations designed.   In view of this,  in recent years
the criterion of entanglement for indistinguishable particles has been studied intensively by several groups \cite{SCKLL01, PY01, LZLL01, ESBL02, GMW02, GM03, WV03, GM04}.    Among them is an algorithmic approach \cite{SCKLL01, LZLL01, ESBL02}, where (anti-)symmetric $N$-particle states separable into $N$ single-particle states are examined based on the Slater rank which is defined from the orthonormal decomposition of the states, generalizing the Schmidt numbers considered for distinguishable particles \cite{Peres93, NC00}.   A related analysis in terms of the von Neumann entropy has also been made for the bipartition separability criterion \cite{PY01, WV03, GM04}.
Although this approach has an advantage in the operational directness, enabling one to construct useful operators such as the entanglement witness, it leaves the extension to the general case of separability into arbitrary subsystems highly non-trivial.
Another approach, which has been proposed by Ghirardi {\it et.\,al} in \cite{GM04, GMW02, GM03}, characterizes separable states based on the possession of \lq complete set of the properties\rq\ associated with the subsystems into which they are separable.  Their criterion of  bipartition separability (an $N$-particle system separable into two subsystems) and full separability (an $N$-particle system separable into $N$ subsystems), given by conditions using projection operators, are satisfied by (anti-)symmetrized direct product states.  For fermionic systems, 
the separable (non-entangled) states determined by these two approaches coincide completely in the case of $N = 2$ bipartition separability \cite{GM04}, which suggests that the two approaches are consistent at least for the restricted cases of fermionic states,  even though their manipulations are quite different.  

In this article, we present a general framework of separability for the case of indistinguishable fermionic particles and furnish a criterion capable of treating 
arbitrary separability, {\it i.e.}, an $N$-particle system separable into an arbitrary number $s$ of subsystems.  The states which meet our criterion are found to be anti-symmetrized states of product states belonging to subsystems whose state spaces are mutually orthogonal, and for the two extreme cases, $s = 2$ and $s = N$, they coincide with those found in \cite{GMW02}. 
It is confirmed that these states have factorizable correlations (as well as weak values \cite{AAV88, AR05, YYKI09}) for a certain class of observables in a way which ensures the cluster decomposition property when the fermions are localized remotely from each other.  
Our argument highlights the important fact that separability of states is defined only with reference to a class of observables admitting an orthogonal structure.  In the actual experimental setups, such a class may correspond to measurements which allow for unambiguous specification of the subsystems in the separation.  

This paper is organized as follows. 
Sec.~II illustrates our motivation by the simple $N = 2$ case from physical grounds, where we discuss why indistinguishability of particles calls for an appropriate revision on the criterion of separability of states.   The novel aspects pertinent to fermionic states and their consequences in the required revision
will be mentioned for our direct generalization of later sections.  
In Sec.~III, we introduce a number of basic tools that we use for the system of general $N$ particles partitioned into arbitrary subsystems.  Two types of observables appropriate for fermionic states, one considered in \cite{GMW02} and the other considered here, are shown to be equivalent at the level of expectation values.   
Sec.~IV contains our main argument on arbitrary
separability.   After introducing an orthogonal structure with which a particular set of states and a class of observables are defined, we present our criterion of arbitrary separability and prove that these sets of states indeed meet the criterion and possess factorizable correlations for the corresponding observables.  
Some of the notable points associated with the orthogonal structure,  as well as the observation on the cluster decomposition, are mentioned in Sec.~V.  
Sec.~VI is devoted to the conclusions and discussions.

\section{Prolegomenon}
 \setcounter{equation}{0}  
 
In this section, we outline our motivation of the present work by discussing the $N = 2$ case referring to the relation between separable states, characterized by factorization of correlations, and the local realism signified, {\it e.g.},  by the Bell inequality \cite{Bellone, CHSH,ADR82,WJSWZ98}.
We shall see that for indistinguishable particles the condition of separability must be revised properly in order to maintain the relation of the two, and this revision will be generalized in later sections when we deal with systems of $N \ge 2$ fermions allowing for arbitrary separations into subsystems.

First, we consider the case of distinguishable particles.  Let $\ket{\Psi} \in {\cal H}_1\ot{\cal H}_2$ be a vector describing the state of the composite system of two particles $i=1,2$ possessing the spaces ${\cal H}_i$ for individual states.    Given the observables $O_i$ for the particles $i=1,2$, the combined observable for the composite system reads $O_1\ot O_2$, which admits the trivial identity,
\be
 O_1\otimes O_2 = \(O_1\otimes \mathbbm{1}_2\)\(\mathbbm{1}_1\otimes O_2\)
 \label{disdec}
\ee
with $\mathbbm{1}_i$ being the identity operator in ${\cal H}_i$.    Corresponding to this, one may find that, for some $\ket{\Psi}$ and $O_i$, the expectation value also factorizes:
\be
\bra{\Psi}O_1\otimes O_2\ket{\Psi}=
   \bra{\Psi}O_1\otimes \mathbbm{1}_2\ket{\Psi}
    \bra{\Psi}\mathbbm{1}_1\otimes O_2\ket{\Psi}.
    \label{inddec}
\ee
As shown in \cite{GMW02}, for pure states of distinguishable particles, the factorization (\ref{inddec}) occurs for any $O_i$ if and only if the state $\ket{\Psi}$ is a product state 
\be
\ket{\Psi} = \ket{\psi_1}_1\ket{\psi_2}_2,
 \label{prodtwo}
\ee
where $\ket{\psi_i}_i \in {\cal H}_i$ are arbitrary states of the individual particles.
This implies that we can use the factorization of correlations (\ref{inddec}) as a criterion to examine the non-entanglement of states of composite systems, if the latter is characterized by the form of product states.

To see the connection between separability and local realism, let us consider the case where  $O_1$, $O_2$ correspond to dichotomic observables $A, B(=\pm1)$ in the constituent particle systems which are remotely separated from each other (see Fig.\ref{bell01}).   
Now, in local realistic theory (LRT) \cite{Bellone}, it is assumed that the outcomes of the measurement is determined by hidden variables collectively denoted by $\lambda$, and that locality of the two particles and the measurement setup ensures that the outcomes are independent of the choice of the local measurement parameters $\vect{a}, \vect{b}$ and also of the outcomes themselves, {\it i.e.}, they are 
determined as $A(\vect{a}, \lambda)$ and $B(\vect{b}, \lambda)$, respectively.    From these, 
the correlations of the measurement outcomes are given by
\be
C_L(\vect{a},\vect{b})=\int\dd\lambda\, \rho(\lambda)A(\vect{a},\lambda)B(\vect{b},\lambda),
 \label{lrtcol}
\ee
under a given probability distribution $\rho(\lambda)$ of hidden variables specifying the state of the system.
We then have the celebrated Bell (CHSH) inequality \cite{CHSH},
\be
|C_L(\vect{a},\vect{b})+C_L(\vect{a},\vect{b}^\prime)+C_L(\vect{a}^\prime,\vect{b})-C_L(\vect{a}^\prime,\vect{b}^\prime)|\leq 2,
 \label{Belleq}
\ee
which holds for any choice of parameters $\vect{a}, \vect{a}^\prime, \vect{b}$ and $\vect{b}^\prime$.  We mention that the Bell inequality (\ref{Belleq}) holds even for stochastic cases \cite{Belltwo} where the dichotomic observables $A, B$ are replaced by \lq averages\rq\ $\bar A, \bar B$ for which we have $\vert \bar A\vert \le 1$ and $\vert \bar B\vert \le 1$.

On the other hand, in quantum mechanics, the dichotomic observable can be represented in the state space ${\cal H}_i = \C^2$ by the spin operators 
$O_i = \sigma_i(\vect{a}):= \sigma_i\cdot \vect{a}$ where $\sigma_i = (\sigma_i^x,\sigma_i^y,\sigma_i^z)$ are Pauli matrices.  Here, the outcomes of the spin measurement in the direction $\vect{a}$ are given by the dichotomic values $\pm1$ as required.   The correlations of the outcomes then read
\be
C_Q(\vect{a},\vect{b})= \bra{\Psi}O_1\otimes O_2\ket{\Psi} = \bra{\Psi}  \sigma_1(\vect{a})\ot\sigma_2(\vect{b}) \ket{\Psi}
 \label{qmcol}
\ee
for which we have the Cirel\rq son inequality \cite{Cirelson80},
 \be
|C_Q(\vect{a},\vect{b})+C_Q(\vect{a},\vect{b}^\prime)+C_Q(\vect{a}^\prime,\vect{b})-C_Q(\vect{a}^\prime,\vect{b}^\prime)|\leq2\sqrt{2}.
 \label{Cir}
 \ee
The upper bound of the Cirel\rq son inequality (\ref{Cir}) is attained by the singlet state $\ket{\Psi} = \ket{{\rm Singlet}}_{12}$ given by
\be
\label{singlet0}
\ket{{\rm Singlet}}_{12}:=\frac{1}{\sqrt{2}}\(\ket{+}_1\ket{-}_2-\ket{-}_1\ket{+}_2\)
\ee
with the eigenstates $\sigma_i^z \ket{\pm}_i = \pm \ket{\pm}_i$ under 
a suitable choice of the parameters (configurations of the measurement axes).  This shows that some quantum states, though not all, do realize correlations that cannot be achieved by LRT.  

If, in particular,  the state $\ket{\Psi}$ is a product state (\ref{prodtwo}), the factorization (\ref{inddec}) yields 
\be
C_Q(\vect{a},\vect{b})=
 \bra{\Psi} \sigma_1(\vect{a})\ot \mathbbm{1}_2\ket{\Psi}
 \bra{\Psi} \mathbbm{1}_1\ot \sigma_2(\vect{b})\ket{\Psi}
 =
{}_1\bra{\psi_1} \sigma_1(\vect{a})\ket{\psi_1}_1
{}_2\bra{\psi_2} \sigma_2(\vect{b})\ket{\psi_2}_2.
 \label{decqubits}
\ee
We then observe that, if we put $\bar A(\vect{a}) = {}_1\bra{\psi_1} \sigma_1(\vect{a})\ket{\psi_1}_1$, $\bar B(\vect{b}) = {}_2\bra{\psi_2} \sigma_2(\vect{b})\ket{\psi_2}_2$ we have $\vert \bar A(\vect{a}) \vert \le 1$ and $\vert \bar A(\vect{a}) \vert \le 1$, and hence the correlation $C_Q(\vect{a},\vect{b})$ in (\ref{decqubits}) coincides formally 
with the stochastic version of $C_L(\vect{a},\vect{b})$ in (\ref{lrtcol}) when the observables are independent of the hidden variables.  Thus, for those factorized states, the upper bound in the inequality (\ref{Cir}) is found to be $2$, which implies that the correlations of factorizable states can always be emulated by LRT.  To sum up, we find that 
for distinguishable particles in pure states, the three criteria (factorization of correlations, direct-productness of state, and emulation by LRT) are all consistent.

\begin{figure}[t]
 \begin{center}
  \includegraphics[height=3cm]{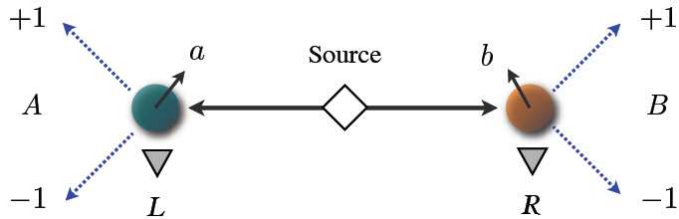}
  \caption{Schematic diagram of experiments to test the Bell inequality.  In a system of distinguishable particles 1 and 2, dichotomic observables $A, B$ are measured along the axes $\vect{a}, \vect{b}$ represented by the operators $ \sigma_1(\vect{a})$, $ \sigma_2(\vect{b})$, respectively.  
Accompanied to these,  apparatuses to measure the positions $L, R$ of the particles represented by the observables $\ket{L}_1{}_1\bra{L}, \ket{R}_2{}_2\bra{R}$ are also installed.  If the particles are indistinguishable, the entire observables must be made invariant under the exchange of the particles in accordance with states which respect the (anti-)symmetry required by the statistics of the particles.}
\label{bell01}
\end{center}
\end{figure}

At this point we note that in the actual measurement we identify the measurement outcomes $A, B$ with those associated with the particles $i = 1, 2$, respectively.   For this, it is customary to consider the spatial degrees of freedom of the particles in addition to the spins, that is, we introduce in each of the constituent systems a state  $\ket{L}_i$ localized around the measurement apparatus on the left and similarly a state $\ket{R}_i$ localized around the measurement apparatus on the right, where the locality implies ${}_i\bra L {R}\rangle_i = 0$.  
With these, the observable and the singlet state of the composite system are given, more precisely, by
\be
 \sigma_1(\vect{a})\ket{L}_1{}_1\bra{L}\otimes \sigma_2(\vect{b})\ket{R}_2{}_2\bra{R}
 \qquad
 {\rm and}
 \qquad
\ket{{\rm Singlet}}_{12}\ket{L}_1\ket{R}_2.
\label{disrep}
\ee
{}Fortunately,  thanks to the direct product structure of the states and the observables (\ref{disrep}), the introduction of the spatial degrees of freedom does not affect the outcomes of correlations, and our argument goes through without invoking them for distinguishable particles.

However, for indistinguishable particles, this is no longer the case.  To see this, we first note that due to the indistinguishability the observable 
for the composite system must also be made invariant under the exchange of particles.  
Thus, the observable corresponding to (\ref{disrep}) should be replaced by 
\begin{equation}
 \sigma_1(\vect{a})\ket{L}_1{}_1\bra{L}\otimes \sigma_2(\vect{b})\ket{R}_2{}_2\bra{R}
  \rightarrow
   \sigma_1(\vect{a})\ket{L}_1{}_1\bra{L}\otimes \sigma_2(\vect{b})\ket{R}_2{}_2\bra{R} 
   + \sigma_1(\vect{b})\ket{R}\bra{R}\otimes \sigma_2(\vect{a})\ket{L}\bra{L}.
  \label{extend1}
\end{equation}
In the case of fermions, one can show (see Proposition 2) that
(\ref{extend1}) is equivalent to 
\begin{equation}
 \sigma_1(\vect{a})\ket{L}_1{}_1\bra{L}\otimes \sigma_2(\vect{b})\ket{R}_2{}_2\bra{R}
  \rightarrow
   A  (\sigma_1(\vect{a})\ket{L}_1{}_1\bra{L}\otimes \sigma_2(\vect{b})\ket{R}_2{}_2\bra{R}) A,
  \label{extend2}
\end{equation}
where $A$ is the (rescaled) anti-symmetrizer
\be
 A = \frac{1}{\sqrt{2}}(\mathbbm{1}_1\ot\mathbbm{1}_2-\pi_{12})
\ee
given from the permutation operator $\pi_{12}$ defined by
\be
 \pi_{12}\ket{\psi_1}_1\ket{\psi_2}_2=\ket{\psi_2}_1\ket{\psi_1}_2.
\ee
{}For notational conciseness, hereafter we use the second form (\ref{extend2}) for observables of fermionic systems.

In order to address the question of separation of correlations, we need to provide beforehand a possible form of separation of operators pertaining to each of the fermions analogously to (\ref{disdec}) in the distinguishable case.   Clearly, the problem is that, despite the indistinguishable nature of the particles, we need to somehow label the particles by the observables in the measurement in order to define the correlations, and one possible approach for this is to utilize the locality which is also presupposed for the distinguishable case.  In the present situation, we have one observable measured by an apparatus on the left and, remotely separated from it, we have another observable measured by an apparatus on the right. 
These are represented by 
\be
 O_1:=\sigma_1(\vect{a})\ket{L}_1{}_1\bra{L},
\qquad
 O_2:=\sigma_2(\vect{b})\ket{R}_2{}_2\bra{R},
\ee
with which the observable of the composite system in (\ref{extend2}) becomes simply $A  (O_1\otimes O_2) A$.  
One can show (see Corollary 2) that it admits the separation with respect to the local observables: 
\be
  A  (O_1\otimes O_2) A=
 A  \(O_1\otimes \mathbb{I}_2\)A \cdot
  A \( \mathbb{I}_1\otimes O_2 \)A.
 \label{decOLR}
\ee
Here we have introduced the projectors,
\be
\mathbb{I}_1 :=I_1 \ket{L}_1{}_1\bra{L},
\qquad
\mathbb{I}_2 :=I_2 \ket{R}_2{}_2\bra{R},
  \label{identabbr}
\ee
with $I_i$ being the identity matrix in the spin space of particle $i$.   Operationally, each of the  projectors (\ref{identabbr}) may be interpreted as the observable which confirms merely the presence of the particles at the apparatus without measuring the spin. 

\begin{table}[t]
\label{rel1}
\begin{center}
\renewcommand\arraystretch{1.5}
\begin{tabular}{c|cc}
 \hline
 \phantom{xxxxxx}  \phantom{xxxxxxxxxxx}  \phantom{xxxxxxxxxx}  & \phantom{xxxxxx} Distinguishable particles  \phantom{xxxxxx} &  \phantom{xxxxxx} Fermions  \phantom{xxxxxx}  \\ 
 \hline \hline
& $O =O_1\otimes O_2\;$ & $O=A\(O_1\otimes O_2 \)A\;$\\
Observables & 
$O^{(1)}=O_1\otimes \mathbbm{1}_2\;$ & $O^{(1)}=A\(O_1\otimes \mathbb{I}_2 \)A\;$\\
& 
$O^{(2)}=\mathbbm{1}_1\otimes O_2\;$ & $O^{(2)}=A\(\mathbb{I}_1\otimes O_2 \)A\;$\\
Factorization of observables & $O=O^{(1)} O^{(2)}$&$O=O^{(1)} O^{(2)}$\\
Factorization of correlations & $\bra{\Psi}O\ket{\Psi}=\bra{\Psi}O^{(1)}\ket{\Psi}\bra{\Psi}O^{(2)}\ket{\Psi}$
 &
 $\bra{\Psi}O \ket{\Psi}=\bra{\Psi}O^{(1)} \ket{\Psi}\bra{\Psi}O^{(2)}\ket{\Psi}$ \quad \\
Separable states & $\ket{\Psi} = \ket{\psi_1}_1\ket{\psi_2}_2$
 &
 $\ket{\Psi} = A \,\ket{\psi_1}_1\ket{\psi_2}_2$\\
 \hline
\end{tabular} 
\end{center}
\caption{Comparison of observables and separable states between distinguishable particles and fermions.}
\end{table}

Now that we have introduced the factorization (\ref{decOLR}) for observables of fermions analogously to the factorization (\ref{disdec}) for observables of distinguishable particles, we may also define the factorization of correlations by
\be
 \label{sec2_decompose1}
 \bra{\Psi}A\(O_1\otimes O_2\)A\ket{\Psi}=
   \bra{\Psi}A\(O_1\otimes \mathbb{I}_2\)A\ket{\Psi}
    \bra{\Psi}A\(\mathbb{I}_1\otimes O_2\)A\ket{\Psi}.
\ee
One can show (see Proposition 3) that, in analogy with (\ref{inddec}) and (\ref{prodtwo}) in the case of distingushable particles,  the factorization (\ref{sec2_decompose1}) occurs if the state $\ket{\Psi}$ is given by
\be
 \label{separable_state}
 \ket{\Psi}= A \,\ket{\psi_1}_1\ket{\psi_2}_2, \qquad
\ket{\psi_1}_1 = \ket{\alpha}_1\ket{L}_1, \quad \ket{\psi_2}_2 = \ket{\alpha^\prime}_2\ket{R}_2,
\ee
where $\ket{\alpha},\ket{\alpha^\prime}$ are arbitrary spin states.   
We note that, as stated before, such states (\ref{separable_state}) possess a complete set of properties mentioned in Sec.~IV (see Definition 1 and Theorem 2).

The correspondence of observables and separable states between distinguishable particles and fermions is summarized in Table I.  One finds that the transition rule from distinguishable particles to fermions is rather simple, {\it i.e.}, it is accomplished essentially by attaching the antisymmterizer $A$ to the states and also to the observables in the conjugate manner.    In fact, this rule is
applicable even to finding states which saturates the Cirel\rq son inequality with the observables given by (\ref{extend2}), 
where one learns that the upper bound of the inequality is attained by the state
\be
\label{singlet2}
 A\(\ket{{\rm Singlet}}_{12}\ot\ket{L}_1\ket{R}_2\)=\ket{{\rm Singlet}}_{12}\ot\frac{1}{\sqrt{2}}\(\ket{L}_1\ket{R}_2+\ket{R}_1\ket{L}_2\).
\ee

An important point to note in the transition rule is that  the assignment of observables entails introducing an \lq orthogonal structure\rq\ there, that is,  due to the locality property ${}_i\bra L {R}\rangle_i = 0$ the operators $O_1$ and $O_2$ are made to act in orthogonal subspaces when the individual spaces of the states are identified.   More generally, the states $\ket{\psi_1}_1$, $\ket{\psi_2}_2$ of the subsystems appearing in (\ref{singlet2}) must be orthogonal in the same sense.   On physical grounds, one may argue that
this orthogonal structure is required for indistinguishable particles to make a distinction between the measurement outcomes obtained for the two observables, while such a structure is unnecessary for distinguishable particles on account of its intrinsic capability of distinction.   As we shall see in the following sections, these characteristics of separability for indistinguishable (fermionic) states will be seen also for the system of general $N$ particles under an arbitrary partition into $s$ subsystems.  Unlike the present $N=2$ case, however, the general $N$ case may admit more than one orthogonal structure, and we shall also discuss the meaning of this ambiguity with respect to the separability criterion there.

\section{Antisymmetric Spaces and Partitioned Observables}
 \setcounter{equation}{0}  

To begin our discussions on arbitrary separability of $N$-particle fermionic states, we first provide a basic account of
identical $N$-particle systems along with necessary tools to treat the separation into arbitrary subsystems. 

Our total system is formed out of $N$ constituent systems representing the $N$ particles.   All of the constituent systems are identical to each other, and the Hilbert space ${\cal H}_i$ associated with the constituent system $i = 1, \dots, N$ is assumed to be $d$-dimensional, that is, ${\cal H}_i=\mathbb{C}^d$ for all $i$.   We assume $d \ge N$ to treat antisymmetric states.
If we let $\ket{a_i}_i$,  $a_i=0, 1, \cdots, d-1$, be a  complete orthonormal basis in ${\cal H}_i$, then the set of 
the direct product states $\bigotimes_{i=1}^N\ket{a_i}_i$ forms a complete orthonormal basis in the Hilbert space  ${\cal H}$ of the total system defined by the tensor product ${\cal H}^{\rm tot} =\bigotimes_{i=1}^N{\cal H}_i$. 

{}For the case of identical particles, a subsystem is a set of constituent systems in arbitrary number.  To specify how the total system is decomposed into such subsystems, 
it is convenient to introduce a
partition $\Gamma$ of the set of labels $\{1,2,\cdots,N\}$ of the constituents, {\it i.e.}, 
\be
\Gamma=\left\{\Gamma_1, \ldots, \Gamma_s \,\, \bigg\vert \, 
\bigcup_{k=1}^s\Gamma_k=\{1,2,\cdots,N\}, 
\quad \Gamma_k\neq \emptyset,
\quad
\Gamma_j\cap \Gamma_k=\emptyset, \quad
j\neq k
\right\}.
\label{defsetone}
\ee
Given a partition $\Gamma$, a subset is specified by 
those constituent systems whose labels comprise $\Gamma_k$ in $\Gamma$ for some $k$, and we use $\Gamma_k$ to denote the subsystem as well. 
The Hilbert space of subsystem $\Gamma_k$ is then given by the tensor product 
\be
{\cal H}(\Gamma_k)=\Bot_{i\in\Gamma_k}{\cal H}_i,
\label{subhsp}
\ee
which is equipped with a complete orthonormal basis $\Bot_{i\in\Gamma_k}\ket{a_i}_i$ in ${\cal H}(\Gamma_k)$. 
If we let $O_{k}$ be an observable in ${\cal H}(\Gamma_k)$, it admits the spectral decomposition,
\be
O_{k}=\sum_{\mu}\lambda_{\mu}\ket{\mu}_{\Gamma_k}{_{\Gamma_k}}\bra{\mu},
\qquad
\lambda_{\mu}\in\mathbb{R},
\qquad
\ket{\mu}_{\Gamma_k} \in {\cal H}(\Gamma_k),
\label{defomegaone}
\ee
fulfilling  
\be
{_{\Gamma_k}}\langle\mu|\nu\rangle_{\Gamma_k}=\delta_{\mu\nu},
\qquad
\ket{\mu}_{\Gamma_k}\in {\cal H}(\Gamma_k),
\qquad
\mu,\,  \nu = 0, 1, \ldots, d^{|\Gamma_k|}-1.
\label{defomegatwo}
\ee
Here, $\delta_{\mu\nu}$ is the Kronecker delta and $|\Gamma_k|$ is the cardinality of the set $\Gamma_k$, which in our case is the number  of constituents contained in the subsystem $\Gamma_k$.

In order to deal with identical fermionic particles, we first introduce  an isomorphic map among the constituent Hilbert spaces ${\cal H}_i$ by means of the correspondence between the basis vectors $\ket{a_i}_i\in{\cal H}_i$ and $\ket{a_i}_j\in{\cal H}_j$ for all $i, j = 1, \ldots, N$. 
The isomorphism is highly non-unique because of the arbitrariness in the choice of bases, but the following arguments hold under any
fixed choice of bases.
One then introduces permutations $\sigma \in S_N$ of the symmetric group $S_N$ on the states of constituents, which are implemented by self-adjoint operators $\pi_\sigma$ on ${\cal H}^{\rm tot}$ defined by
\be
\pi_\sigma\(\bigotimes_{i=1}^N\ket{a_i}_i\)=\bigotimes_{i=1}^N\ket{a_i}_{\sigma(i)}.
\label{defperm}
\ee
Note that $\pi_\sigma$ acts on the bra vectors as
\be
\(\bigotimes_{i=1}^N {}_i\bra{a_i}\)\pi_\sigma=\bigotimes_{i=1}^N{}_{\sigma(i)}\bra{a_i},
\ee
since $\pi_\sigma$ are self-adjoint. 
The action of permutation $\sigma$ on the labels of the constituents induces a map $\Gamma$ into another partition 
$\Gamma^\sigma=\{\Gamma_k^\sigma\}_{k=1}^s$ with $\Gamma_k^\sigma=\{\sigma(i)\,|\, i\in\Gamma_k\}$. 
Accordingly, we may classify permutations into two classes;  (i) those leaving the partition $\Gamma$ unchanged, {\it i.e.}, $\Gamma^\sigma=\Gamma$, and (ii) those changing $\Gamma$,  {\it i.e.}, $\Gamma^\sigma\neq\Gamma$.  We denote the set of 
permutations belonging to the former class by ${\cal I}(\Gamma)$ and the set of the latter by $\bar{\cal I}(\Gamma)$.   
The symmetric group $S_N$ is thus a disjoint union of these, $S_N = {\cal I}(\Gamma) \cup \bar{\cal I}(\Gamma)$, and we may further
consider the quotient set $S_N/{\cal I}(\Gamma)$ of $S_N$ with respect to ${\cal I}(\Gamma)$.

The anti-symmetrizer ${\cal A}$, which plays a key role in the subsequent discussions of fermionic states,  is constructed 
from the permutation operators as 
\be
\label{asym01}
      {\cal A}:= {1\over{N !}} \sum_{\sigma\in S_N}{\rm sgn}(\sigma) \pi_{\sigma},
      \label{defAcal}
\ee
where ${\rm sgn}(\sigma)=1$ for even permutations, and ${\rm sgn}(\sigma)=-1$ for odd ones. 
Note that the anti-symmetrizer ${\cal A}$ is a projection operator, satisfying $\mathcal{A}^2=\mathcal{A}$, $\mathcal{A}^\dag= \mathcal{A}$ and
\be
\pi_\sigma \mathcal{A}=\mathcal{A}\pi_\sigma={\rm sgn}(\sigma)\mathcal{A}
\label{Api}
\ee
for all $\sigma\in S_N$. 

The space of fermionic states ${\cal H}_A^{\rm tot}$ is the antisymmetric linear subspace of ${\cal H}^{\rm tot}$ obtained by
collecting all antisymmetric states  in  ${\cal H}^{\rm tot}$, which is written as
\be
{\cal H}_A^{\rm tot}= {\rm Asym}\,{\cal H}^{\rm tot}
:= \big\{\ket{\Psi}\,\big|\,\mathcal{A}\ket{\Psi}=\ket{\Psi},\, 
\ket{\Psi}\in{\cal H}^{\rm tot} \big\}.
\label{defasym}
\ee
Likewise, we may construct the antisymmetric subspace of a given multi-particle space ${\cal K}$ by considering only those $\pi_\sigma$ corresponding to the permutations $\sigma \in S_{|\Gamma_k|}$ which map vectors from ${\cal K}$ to ${\cal K}$. 
Utilizing these permutation operators, we can define the anti-symmetrizer $\mathcal{A}$ on ${\cal K}$ by (\ref{defAcal}) upon restriction to the permutations appropriate for ${\cal K}$ (we use the same symbol $\mathcal{A}$ for simplicity), and thereby define the antisymmetric subspace 
$ {\rm Asym}\,{\cal K}$ analogously to (\ref{defasym}) from ${\cal K}$.    For example, if this construction is applied to the space ${\cal H}(\Gamma_k)$ in (\ref{subhsp}), we then have the Hilbert space ${\cal H}_A(\Gamma_k) := {\rm Asym}\, {\cal H}(\Gamma_k)$ of the subsystem $\Gamma_k$ consisting of only antisymmetric states.

Having furnished necessary tools for our discussion, we now discuss observables in ${\cal H}_A^{\rm tot}$ constructed out of observables on subsystems specified by the partition $\Gamma$.   Given an observable $O_{k}$ in ${\cal H}_A(\Gamma_k)$ for each $k = 1, \ldots, s$, we may consider, by taking account of the anti-symmetric property of the states required, the operator in ${\cal H}_A^{\rm tot}$,
\be
O(\Gamma) := A(\Gamma)\(\bigotimes_{k=1}^sO_{k}\)A(\Gamma),
\label{defO}
\ee
where we have introduced, for our technical convenience, the rescaled anti-symmetrizer,
\be
\label{resasym}
      A(\Gamma):= \sqrt{M(\Gamma)}\, {\cal A}
      \label{defA}
\ee
with the multinomial coefficient associated with the partition $\Gamma$,
\be
\label{cardm}
M(\Gamma):= {{N!}\over{\prod_{k=1}^s|\Gamma_k|!}}.
\ee
We denote the class of operators associated 
with the partition $\Gamma$ of the form (\ref{defO}) by ${\cal C}(\Gamma)$.  
A salient property of the class ${\cal C}(\Gamma)$ is that it is invariant under permutation, ${\cal C}(\Gamma^\sigma)= {\cal C}(\Gamma)$ for any $\sigma \in S_N$.  This can be seen from
\begin{proposition}
To an observable $O(\Gamma) \in {\cal C}(\Gamma)$ defined in (\ref{defO}), one can always find an observable $O(\Gamma^\sigma) \in {\cal C}(\Gamma^\sigma)$ for any $\sigma \in S_N$ such that 
\be
O(\Gamma)=O(\Gamma^\sigma).
\ee
\end{proposition}
\begin{proof}
Using (\ref{Api}) and the spectral decomposition 
of $O_{k}$ with the label $\mu$ written here as $\mu_k$ to distinguish the subsystems the states belong to, we find
\be
\begin{split}
\label{propr}
O(\Gamma)&= A(\Gamma)\(\Bot_{k=1}^s\sum_{\mu_k}\lambda_{\mu_k}\ket{\mu_k}_{\Gamma_k} {_{\Gamma_k}}\bra{\mu_k}\)A(\Gamma)\\
&=\sum_{\mu_1,\cdots,\mu_s}\(\prod_{k=1}^s\lambda_{\mu_k}\)A(\Gamma)\(\Bot_{k=1}^s\ket{\mu_k}_{\Gamma_k}\)\(\Bot_{k=1}^s {_{\Gamma_k}}\bra{\mu_k}\)A(\Gamma)\\
&=\sum_{\mu_1,\cdots,\mu_s}\(\prod_{k=1}^s\lambda_{\mu_k}\)A(\Gamma)\pi_\sigma\(\Bot_{k=1}^s\ket{\mu_k}_{\Gamma_k}\)\(\Bot_{k=1}^s {_{\Gamma_k}}\bra{\mu_k}\)\pi_\sigma A(\Gamma)\\
&=\sum_{\mu_1,\cdots,\mu_s}\(\prod_{k=1}^s\lambda_{\mu_k}\)A(\Gamma^\sigma)\(\Bot_{k=1}^s\ket{\mu_k}_{\Gamma_k^\sigma}\)\(\Bot_{k=1}^s {_{\Gamma_k^\sigma}}\bra{\mu_k}\)A(\Gamma^\sigma)\\
&=A(\Gamma^\sigma)\(\Bot_{k=1}^s\sum_{\mu_k}\lambda_{\mu_k}\ket{\mu_k}_{\Gamma_k^\sigma} {_{\Gamma_k^\sigma}}\bra{\mu_k}\)A(\Gamma^\sigma)
\end{split}
\ee
where we have used $A(\Gamma)=A(\Gamma^\sigma)$.   The last line gives an observable belonging to ${\cal C}(\Gamma^\sigma)$, and hence if we just define it to be our $O(\Gamma^\sigma)$, we have $O(\Gamma)=O(\Gamma^\sigma)$.  
\end{proof}

\medskip

The above proposition shows that, as far as the observables belonging to the class ${\cal C}(\Gamma)$ are concerned, any partitions which are related under permutations lead to the same outcomes in the total system ${\cal H}_A^{\rm tot}$.
On the other hand, one may instead consider a different class $\tilde{\cal C}(\Gamma)$ of observables constructed from the same set of observables $O_{k}$ in ${\cal H}_A(\Gamma_k)$  by collecting all possible permutations of the product,
\be
{\tilde O}(\Gamma):=\sum_{\sigma\in S_N/{\cal I}(\Gamma)} \pi_\sigma \left( \Bot_{k=1}^s O_{k} \right) \pi_\sigma,
\label{defcalO}
\ee
where we take the summation over the quotient set $S_N/{\cal I}(\Gamma)$ to eliminate
possible degeneracies due to the invariant permutations in ${\cal I}(\Gamma)$.
The two observables,  $O(\Gamma)$ in (\ref{defO}) and ${\tilde O}(\Gamma)$ in (\ref{defcalO}),  enjoy a simple relationship as seen in \begin{lemma}
The observables $O(\Gamma)$ and ${\tilde O}(\Gamma)$ defined in (\ref{defO}) and  (\ref{defcalO}), respectively, are related by 
\be
\mathcal{A}\, {\tilde O}(\Gamma)\, \mathcal{A}=O(\Gamma).
\ee
\end{lemma}
\begin{proof}
Using (\ref{Api}) and noting the cardinality of $S_N/{\cal I}(\Gamma)$ being $M(\Gamma)$, we have
\be
\mathcal{A}\, {\tilde O}(\Gamma)\, \mathcal{A}=\sum_{\sigma\in S_N/{\cal I}(\Gamma)}\mathcal{A} \pi_\sigma \left( \Bot_{k=1}^s O_{k} \right) \pi_\sigma \mathcal{A}=\sum_{\sigma\in S_N/{\cal I}(\Gamma)}\frac{O(\Gamma)}{M(\Gamma)}=O(\Gamma).
\ee
\end{proof}

\medskip

Now, it is straightforward to show
\begin{proposition}
Two observables $O(\Gamma) \in {\cal C}(\Gamma)$ and ${\tilde O}(\Gamma) \in \tilde{\cal C}(\Gamma)$ which are constructed from a given set of observables $\{O_{k}\}_{k=1}^s$ by the prescriptions (\ref{defO}) and (\ref{defcalO}), respectively, share the same expectation values under arbitrary antisymmetric states $\ket{\Psi}\in{\cal H}_A^{\rm tot}$, that is,
\be
\label{prop2}
\bra{\Psi}O(\Gamma)\ket{\Psi}=\bra{\Psi}{\tilde O}(\Gamma)\ket{\Psi},
\qquad \forall \ket{\Psi}\in{\cal H}_A^{\rm tot}.
\ee
\end{proposition}
\begin{proof}
Using Lemma 1 and $\mathcal{A}\ket{\Psi} = \ket{\Psi}$ for $\ket{\Psi}\in{\cal H}_A^{\rm tot}$, we obtain (\ref{prop2}).
\end{proof}

\medskip

At this point we briefly consider the problem of constructing observables of $N$-particle systems from those of single-particle constituent systems.  For simplicity,
we use the abbreviated notation $A:=A(\Gamma)$ hereafter.   To examine the physical meaning of the observable $O(\Gamma)$ defined in (\ref{defO}), let us consider the $N=2$ case $\Gamma=\{\Gamma_1,\Gamma_2\}$ with $\Gamma_1=\{1\}$ and $\Gamma_2=\{2\}$, where we have observables $O_1, O_1^\prime$ in ${\cal H}_1$ and $O_2, O_2^\prime$ in ${\cal H}_2$.  We assume that $O_1$ and $O_2$ (and similarly $O_1^\prime$ and $O_2^\prime$) are operators sharing the same spectrum decomposition with the same eigenvectors in each of the constituent Hilbert space.   In other words, we just have two sets of physically distinct operators acting on the spaces, ${\cal H}_1$ and ${\cal H}_2$, and accordingly we may drop the indices $1$ and $2$ which indicate the labels of the subsystems that the operators act on and write them simply as $O$ and $O^\prime$.  
The observable $O(\Gamma)$ acting on the total system ${\cal H}_A^{\rm tot}$ constructed from the two observables $O$ and $O^\prime$ would then be
\be
\label{defOroundzero}
O(\Gamma)=A\(O\ot O^\prime\)A.
\ee
We are thus tempted to interpret $O(\Gamma)$ as an operator implementing a \lq simultaneous\rq~measurement of $O$ and $O^\prime$ on the constituent systems.  
On the other hand, we may expect that the individual measurements of the two observables on the $N=2$ system are represented by 
\be
Q(\Gamma)=A\(O\ot\mathbbm{1}\)A,
\qquad
{\rm and}
\qquad
Q^\prime(\Gamma)=A\(O^\prime\ot\mathbbm{1}\)A,
\label{defOroundone}
\ee
respectively.
We observe that, assuming that the two observables commute $[O,O^\prime]=0$, 
the product of the two operators in (\ref{defOroundone}) becomes a self-adjoint operator and reads 
\be
Q(\Gamma)\, Q^\prime(\Gamma)=O(\Gamma)+A\(OO^\prime\ot\mathbbm{1}\)A,
\ee
where $O(\Gamma)$ is the one given in (\ref{defOroundzero}).
Imposing further that the two operators be orthogonal (or their
supports be mutually disjoint) $OO^\prime=0$, we obtain
\be
Q(\Gamma)\, Q^\prime(\Gamma)=O(\Gamma).
\ee
When this is done, the final form admits the interpretation that the product of two individual measurements yields the simultaneous measurement given by the product form (\ref{defOroundzero}).

Our argument can be easily extended to the case of infinite dimensional Hilbert spaces such as ${\cal H}_i=\C^d\ot L^2(\R^3)$, where one may consider two spatially localized observables possessing no overlaps.  Such an example has been mentioned in
\cite{GMW02} in considering separability for bipartition
of $N$-particle states, where observables of the class $\tilde{\cal C}(\Gamma)$ which fulfill the orthogonality $OO^\prime=0$ are adopted.
Our discussion to be expounded below will be based on the observables of the class ${\cal C}(\Gamma)$, but the consistency between the two is ensured by Proposition 2 which asserts that at the level of expectation values for $\ket{\Psi}\in{\cal H}_A^{\rm tot}$, the choice of the class, 
${\cal C}(\Gamma)$ or $\tilde{\cal C}(\Gamma)$, is immaterial.    More on this point will be discussed in Sec.~V where we consider the physical implications of separability.

\section{Arbitrary Separability for $N$-Particle Antisymmetric  States}
 \setcounter{equation}{0}  

We now move on to present our condition for separability of $N$-particle states into arbitrary subsystems and provide states which fulfill the condition.  To this end, following the procedure employed in \cite{GMW02}, and also hinted by the argument given at the end of the previous section, we furnish 
a set of orthogonal subspaces within $\C^d$ and thereby obtain operators which are mutually orthogonal and admit their simultaneous measurement.   This will allow us to find a set of states in  ${\cal H}_A^{\rm tot}$ which are factorizable with respect to a given partition $\Gamma$ of the original system under the orthogonal structure.

\subsection{Orthogonal Structure in States and Observables}

Let us first choose a set $V$ of mutually orthogonal linear subspaces $V_j$, $j = 1, \ldots, m$, in $ \C^d$, {\it i.e.},
\be
\label{defpart}
V=\left\{V_1, \ldots, V_m \,\, \Bigg\vert \, 
\bigoplus_{k=1}^m V_k\subseteq\C^d, \quad
V_k\bot V_l, \quad
k\neq l
\right\}.
\label{defset}
\ee
Such a set $V$ can be used to furnish an orthogonal structure in each of the constituent Hilbert spaces, by
considering the subspaces $V_k$ within ${\cal H}_i = \C^d$ and denote them as $V_k[{\cal H}_i] \subset {\cal H}_i$ for
$k = 1, \ldots, m$.   For our purposes, we choose $V$ in (\ref{defset}) such that $m = s$, where $s$ is the number of subsystems in $\Gamma=\{\Gamma_k\}_{k=1}^s$, and that 
$\dim V_k \ge  \vert \Gamma_k \vert$ for all $k$.   We associate a common  subspace $V_k$ to all the constituents belonging to $\Gamma_k$, and then define the subspace $W(\Gamma_k, V_k) \subset {\cal H}_A(\Gamma_k)$ by
\be
\label{wspa}
W(\Gamma_k, V_k) := {\rm Asym}\left(\bigotimes_{i\in\Gamma_k}V_k[{\cal H}_i]\right).
\ee
In more concrete terms, we  
choose those states in ${\cal H}(\Gamma_k)$ which are confined to the subspace $V_k$ in each of the constituent ${\cal H}_i$ for $i \in \Gamma_k$ and then anti-symmetrize them with the procedure analogous to (\ref{defasym}) and the remarks that follow.   The resultant states form the space $W(\Gamma_k, V_k)$ in (\ref{wspa}).

Since $W(\Gamma_k,V_k)$ is itself a Hilbert space and has a complete orthonormal basis, we may consider
observables $O_k$ with its support only in $W(\Gamma_k,V_k) \subset {\cal H}_A(\Gamma_k)$.  In other words,  
$O_k$ admits the spectral decomposition,  
\be
\label{ospecw}
O_{k}=\sum_{\mu}\lambda_{\mu}\ket{\mu}_{\Gamma_k}{_{\Gamma_k}}\bra{\mu},
\qquad
\lambda_{\mu}\in\mathbb{R},
\qquad
\ket{\mu}_{\Gamma_k} \in W(\Gamma_k,V_k).
\ee
Note that $O_k$ is still a self-adjoint operator in ${\cal H}_A(\Gamma_k)$, and 
out of these observables in the subsystems we can construct operators on the form (\ref{defO}) as before.  More explicitly, we consider the class of operators in ${\cal H}_A^{\rm tot}$ given by
\be
\label{clde}
{\cal C}(\Gamma, V):=\left\{O\, \Bigg\vert \,\,  O= A\left(\bigotimes_{k=1}^sO_k\right)A \right\},
\ee
with the requirement (\ref{ospecw}).  Note that ${\cal C}(\Gamma, V)$ is uniquely defined to a given pair $(\Gamma, V)$.
Note also that, 
by construction, ${\cal C}(\Gamma, V)$ is not closed under addition, that is, $O+O^\prime\not\in {\cal C}(\Gamma, V)$ for 
$O$, $O' \in {\cal C}(\Gamma, V)$ in general.
Further, an operator $O$ may belong to more than one such classes defined from different orthogonal sets $V$.  

Similarly, we also consider the set of states which are obtained by anti-symmetrizing fully factorized states with respect to the partition $\Gamma$ 
\be
\label{asdec}
S(\Gamma, V):=\left\{\ket{\Psi}\, \Bigg\vert \,\,  \ket{\Psi}=A\bigotimes_{k=1}^s\ket{\psi_k}, \quad  \ket{\psi_k}\in W(\Gamma_k,V_k)\right\}.
\ee
Note that $S(\Gamma, V) \subseteq {\cal H}_A^{\rm tot}$ and that $S(\Gamma, V)$ is not a linear space.  Note also that, in general, two sets $S(\Gamma, V)$ and $S(\Gamma, V')$ for different $V$ and $V'$ are not mutually exclusive in ${\cal H}_A^{\rm tot}$.  
This implies that a single state may belong both to $S(\Gamma, V)$ and $S(\Gamma, V')$ or possibly more, as we see explicitly in Sect.~V.  Hereafter we verify that the states in $S(\Gamma,V)$ correspond to the separable states of distinguishable particles from the two view points: Decomposition of correlations and complete set of the properties.

\subsection{Factorization of Correlations}

The next lemma is useful in our following arguments.
\begin{lemma}
The inner product of any two states $\ket{\Psi}$, $\ket{\Phi} \in S(\Gamma, V)$ is factorizable according to the partition $\Gamma$,
\be
\braket{\Phi|\Psi}=\prod_{k=1}^s\braket{\phi_k|\psi_k}.
\label{decip}
\ee
In other words, for $\ket{\psi_k}, \ket{\phi_k}\in W(\Gamma_k,V_k)$, $k = 1, \ldots, s$, we have
\be
\(\bigotimes_{k=1}^s\bra{\phi_k} A\)\( A\bigotimes_{k=1}^s\ket{\psi_k}\)=\prod_{k=1}^s\braket{\phi_k|\psi_k}.
\label{dpdec}
\ee
\end{lemma}
\begin{proof}
In terms of the set ${\cal I}(\Gamma)$ of invariant permutations and its complement
$\bar{\cal I}(\Gamma)$, the anti-symmetrizer $\mathcal{A}$ can be decomposed into
\be
\mathcal{A}=F+G, \qquad 
\hbox{with}\qquad 
F:=\frac{1}{N!}\sum_{\sigma\in {\cal I}({\Gamma})}{\rm sgn}(\sigma)\pi_\sigma,
\qquad 
G:=\frac{1}{N!}\sum_{\sigma\in \bar{{\cal I}}({\Gamma})}{\rm sgn}(\sigma)\pi_\sigma.
\ee
{}For permutations $\sigma\in{\cal I}(\Gamma)$, we obtain
\be
\pi_\sigma\bigotimes_{k=1}^s\ket{\psi_k}={\rm sgn}(\sigma)\bigotimes_{k=1}^s\ket{\psi_k},
\ee
which yields
\be
F\bigotimes_{k=1}^s\ket{\psi_k}=\frac{1}{M}\bigotimes_{k=1}^s\ket{\psi_k},
\label{fdec}
\ee
with $M = M(\Gamma)$ in (\ref{cardm}), because the cardinality of  ${\cal I}(\Gamma)$ is equal to $\prod_{k=1}^s|\Gamma_k|!$.
Also, we have
\be
\(\bigotimes_{k=1}^s\bra{\phi_k}\)G\bigotimes_{k=1}^s\ket{\psi_k}=0.
\label{gdec}
\ee
This can be seen by noticing that any permutation $\sigma\in\bar{\cal I}(\Gamma)$ contained in $G$ maps all the states in $V_k[ {\cal H}_i]$ for some $i\in\Gamma_k$ 
to the corresponding states in $V_k[{\cal H}_{\sigma(i)}]$ with $\sigma(i)\not\in\Gamma_k$.  As a result, the l.h.s.~of (\ref{gdec}) vanishes due to the orthogonality $V_k[{\cal H}_{\sigma(i)}]\bot V_l[{\cal H}_{\sigma(i)}]$ for $l$ for which $\Gamma_l\ni\sigma(i)$. 
Using (\ref{fdec}) and (\ref{gdec}), we then see
\be
\begin{split}
\braket{\Phi|\Psi}
&=\(\bigotimes_{k=1}^s\bra{\phi_k} A\)\(A\bigotimes_{k=1}^s\ket{\psi_k}\)\\
 &= M\(\bigotimes_{k=1}^s\bra{\phi_k}\)\mathcal{A}\bigotimes_{k=1}^s\ket{\psi_k}\\
 &= M\(\bigotimes_{k=1}^s\bra{\phi_k}\)F\bigotimes_{k=1}^s\ket{\psi_k}+M\(\bigotimes_{k=1}^s\bra{\phi_k}\)G\bigotimes_{k=1}^s\ket{\psi_k}\\
 &= \prod_{k=1}^s\braket{\phi_k|\psi_k}.
 \end{split}
\ee
\end{proof}

\medskip


\medskip

\medskip

Lemma 2 is also useful to show that the set $S(\Gamma, V)$ is closed under the action of the observables in ${\cal C}(\Gamma, V)$, that is, 
\begin{lemma}
{}For any $\ket{\Psi}\in S(\Gamma,V)$ and $O\in{\cal C}(\Gamma,V)$, we have $O\ket{\Psi}\in S(\Gamma, V)$.  More explicitly, if
$\ket{\Psi}$ and $O$ are given by (\ref{asdec}) and (\ref{clde}), respectively, then 
\be
\label{opsi}
O\ket{\Psi}=A\Bot_{k=1}^s O_k\ket{\psi_k} \in S(\Gamma,V).
\ee
\end{lemma}
\begin{proof}
Employing the spectral decomposition (\ref{ospecw}), we find 
\be
\begin{split}
O\ket{\Psi}&=A\left(\bigotimes_{l=1}^sO_l\right)A\left(A\bigotimes_{m=1}^s\ket{\psi_m}\right)\\
&=A\left(\bigotimes_{l=1}^s\sum_{\mu_l}\lambda_{\mu_l}\ket{\mu_l}\bra{\mu_l}\right)A\left(A\bigotimes_{m=1}^s\ket{\psi_m}\right)\\
&=A\left[\sum_{\mu_1,\cdots,\mu_s}\(\prod_{l=1}^s\lambda_{\mu_l}\)\(\Bot_{l=1}^s\ket{\mu_l}\)\(\Bot_{l=1}^s\bra{\mu_l}\)\right]A\left(A\bigotimes_{m=1}^s\ket{\psi_m}\right).\\
\end{split}
\ee
With the help of  (\ref{dpdec}), and using the spectral decomposition again and changing the dummy indices in summand, we obtain
\be
\begin{split}
O\ket{\Psi}&=A\sum_{\mu_1,\cdots,\mu_s}\(\prod_{l=1}^s\lambda_{\mu_l}\)\(\bigotimes_{k=1}^s\ket{\mu_k}\)
\left[\left(\bigotimes_{k=1}^s\bra{\mu_k}A\right)\(A\bigotimes_{k=1}^s\ket{\psi_k}\)\right]\\
&=A\sum_{\mu_1,\cdots,\mu_s}\(\prod_{l=1}^s\lambda_{\mu_l}\)\bigotimes_{k=1}^s\ket{\mu_k}\braket{\mu_k|\psi_k}\\
&=A\bigotimes_{k=1}^s\sum_{\mu_k}\lambda_{\mu_k}\ket{\mu_k}\braket{\mu_k|\psi_k}\\
&=A\Bot_{k=1}^sO_k\ket{\psi_k}.
\end{split}
\ee
Since $O_k$ has the support only in $W(\Gamma_k, V_k)$, we observe $O_k\ket{\psi_k}\in W(\Gamma_k, V_k)$ for all $k$, and hence $O\ket{\Psi}\in S(\Gamma,V)$.
\end{proof}

\medskip

Factorizability of weak values can then be established by the help of Lemma 2 and Lemma 3 as
\begin{theorem}
The weak values of $O \in {\cal C}(\Gamma, V)$ under $\ket{\Psi}, \ket{\Phi}\in S(\Gamma, V)$ 
is factorizable according to the partition $\Gamma$,
\be
\frac{\bra{\Phi}O\ket{\Psi}}{\braket{\Phi|\Psi}}=\prod_{k=1}^s\frac{\bra{\phi_k}O_k\ket{\psi_k}}{\braket{\phi_k|\psi_k}}.
\label{decweak}
\ee
\end{theorem}
\begin{proof}
The factorization of the denominator of the l.h.s. follows from Lemma 2 for  $\ket{\Psi}, \ket{\Phi}\in S(\Gamma, V)$.  
Since $O\ket{\Psi}$ also belongs to $S(\Gamma, V)$ by Lemma 3, again from Lemma 2 the factorization of the numerator of the l.h.s. follows, too. Combining these we find (\ref{decweak}).
\end{proof}

\medskip

As a trivial corollary of the theorem, we note here that the expectation values of observables $O \in {\cal C}(\Gamma, V)$ are factorizable for states in $S(\Gamma, V)$:
\begin{corollary}
The expectation values of $O \in {\cal C}(\Gamma, V)$ under a normalized state $\ket{\Psi}\in S(\Gamma, V)$ 
is factorizable according to the partition $\Gamma$,
\be
 \bra{\Psi}O\ket{\Psi}=\prod_{k=1}^s\bra{\psi_k}O_k\ket{\psi_k}.
\label{decO}
\ee
\end{corollary}
\medskip

The factorizability of expectation values according to (arbitrary) partitions $\Gamma$ discussed above generalizes the previous result \cite{GMW02} obtained for the case of bipartitions $s = 2$.
The set  ${\cal C}(\Gamma, V)$ of observables for which the factorizability is seen is also an extension of the set used in \cite{GMW02}. 

It is worth mentioning that the space ${\cal C}(\Gamma, V)$ is closed under multiplication provided that the self-adjointness of the operator is preserved.   Namely, we have
\begin{lemma}
Given two elements $O$, $O' \in {\cal C}(\Gamma, V)$ with
\be
\begin{split}
O= A\left(\bigotimes_{k=1}^sO_k\right)A,
\qquad
O^\prime= A\left(\bigotimes_{k=1}^sO^\prime_k\right)A,
\end{split}
\ee
for which the product $O_k O^\prime_k$ of the two
self-adjoint operators $O_k$ and $O^\prime_k$ is again self-adjoint for all $k$, we have
\be
\label{ooprod}
OO^\prime= A\left(\bigotimes_{k=1}^sO_kO_k^\prime\right)A\in {\cal C}(\Gamma, V).
\ee
\end{lemma}
\begin{proof}
With the use of the spectral decomposition, we find
\be
\begin{split}
OO^\prime&= A\left(\bigotimes_{k=1}^sO_k\right)A\cdot A\left(\bigotimes_{l=1}^sO_l^\prime\right)A\\
&=A\(\Bot_{k=1}^s\sum_{\mu_k}\lambda_{\mu_k}\ket{\mu_k}\bra{\mu_k}\)A\cdot A\(\Bot_{k=1}^s\sum_{\nu_k}\lambda_{\nu_k}^\prime\ket{\nu_k}\bra{\nu_k}\)A\\
&=\sum_{\mu_1,\cdots,\mu_s}\sum_{\nu_1,\cdots,\nu_s}\(\prod_{k=1}^s\lambda_{\mu_k}\lambda_{\nu_k}^\prime\)A\(\Bot_{k=1}^s\ket{\mu_k}\)\[\(\Bot_{k=1}^s\bra{\mu_k}A\)\(A\Bot_{k=1}^s\ket{\nu_k}\)\]\(\Bot_{k=1}^s\bra{\nu_k}\)A.
\end{split}
\ee
From (\ref{dpdec}) and using the spectral decomposition again, we obtain
\be
\begin{split}
OO^\prime&=\sum_{\mu_1,\cdots,\mu_s}\sum_{\nu_1,\cdots,\nu_s}\(\prod_{k=1}^s\lambda_{\mu_k}\lambda_{\nu_k}^\prime\)A\(\Bot_{k=1}^s\ket{\mu_k}\)\(\prod_{k=1}^s\braket{\mu_k|\nu_k}\)\(\Bot_{k=1}^s\bra{\nu_k}\)A\\
&=\sum_{\mu_1,\cdots, \mu_s}\sum_{\nu_1,\cdots,\nu_s}\(\prod_{k=1}^s\lambda_{\mu_k}\lambda_{\nu_k}^\prime\)A\(\Bot_{k=1}^s\ket{\mu_k}\)\(\Bot_{k=1}^s\bra{\mu_k}\)\(\Bot_{k=1}^s\ket{\nu_k}\)\(\Bot_{k=1}^s\bra{\nu_k}\)A\\
&=A\(\Bot_{k=1}^s\sum_{\mu_k}\lambda_{\mu_k}\ket{\mu_k}\bra{\mu_k}\)\(\Bot_{k=1}^s\sum_{\nu_k}\lambda_{\nu_k}^\prime\ket{\nu_k}\bra{\nu_k}\)A\\
&=A\left(\bigotimes_{k=1}^sO_kO_k^\prime\right)A.
\end{split}
\ee
Since $O_k$ and $O_k^\prime$ are in the form (\ref{ospecw}), so is the product $O_kO_k^\prime$.  Besides, since the product $O_kO_k^\prime$ is assumed to be self-adjoint for all $k$, we obtain (\ref{ooprod}).
\end{proof}

\medskip

Let $\mathbb{I}_k$ be the projection operator onto $W(\Gamma_k, V_k)$ from ${\cal H}(\Gamma_k)$, {\it i.e.}, the operator (\ref{ospecw}) with $\lambda_\mu = 1$ for all $\mu$, which acts as the identity operator when restricted to $W(\Gamma_k, V_k)$.
Given an $O \in {\cal C}(\Gamma, V)$ in (\ref{clde}) with $O_k$, $k = 1, \ldots, s$, we consider the corresponding observables in ${\cal H}_A^{\rm tot}$ given by
\be
\label{defO(k)}
 O^{(k)}(\Gamma,V) := A\(\bigotimes_{i=1}^{k-1} \mathbb{I}_i\)\otimes O_k \(\bigotimes_{i=k+1}^s \mathbb{I}_i\)A.
\ee
Note that, since $\mathbb{I}_i$ are all self-adjoint, we have $O^{(k)}(\Gamma,V) \in{\cal C}(\Gamma,V)$.   
We then find from Lemma 4 that 
\begin{corollary}
An observable $O\in{\cal C}(\Gamma,V)$ given in  (\ref{clde}) with $O_k$, $k = 1, \ldots, s$, admits the factorization,
\be
O=\prod_{k=1}^s O^{(k)}(\Gamma,V),
\label{decOOk}
\ee
in terms of the observables $O^{(k)}(\Gamma,V)$ for $k = 1, \ldots, s$ defined in (\ref{defO(k)}).
\end{corollary}

\medskip

It is readily seen that for $\ket{\Psi}\in S(\Gamma,V)$ the factorization corresponding to (\ref{decOOk}) also holds at the level of the expectation value:

\begin{proposition}
\label{propdecexp}  
{}For a normalized state $\ket{\Psi}\in S(\Gamma,V)$, the expectation value of an observable $O\in{\cal C}(\Gamma,V)$ in  (\ref{clde}) factorizes as
\be
 \bra{\Psi}O\ket{\Psi}
  =\prod_{k=1}^s \bra{\Psi}O^{(k)}(\Gamma,V)\ket{\Psi},
\label{decOk}
\ee
where $O^{(k)}(\Gamma,V)$, $k = 1, \ldots, s$, are observables given in (\ref{defO(k)}).
\end{proposition}
\begin{proof}
Application of Corollary 1 to the both sides of (\ref{decOk}) yields the same result given by the r.h.s.~of (\ref{decO}). 
\end{proof}
\medskip

\subsection{Criterion of Separability and the Complete Set of Properties}

We now present our criterion of  separability of $N$-particle fermionic systems into arbitrary subsystems:   
\begin{definition}
{\it A normalized $N$-particle antisymmetric pure state $\ket{\Psi}\in{\cal H}_A^{\rm tot}$ is {\rm separable with respect to $(\Gamma, V)$} if there exist 
one-dimensional projection operators $P_k$, $k=1,\cdots,s$,  projecting in $W(\Gamma_k, V_k)$ from ${\cal H}_A(\Gamma_k)$ such that
\be
\bra{\Psi}P^{(k)}(\Gamma, V)\ket{\Psi}=1,
\label{defsep}
\ee
with}
\be
\begin{split}
P^{(k)}(\Gamma, V)
:=A\left(\bigotimes_{i=1}^{k-1}\mathbb{I}_i\right)\otimes P_k\otimes\left(\bigotimes_{i=k+1}^{s}\mathbb{I}_i\right)A.
\end{split}
\label{defPk}
\ee
\end{definition}

\medskip

This definition of separability with respect to $(\Gamma, V)$ is motivated by the fact that it gives a necessary and sufficient condition 
for states to have factorizable inner products among themselves according to the partition $\Gamma$ as well as factorizable expectation values for observables belonging to the class ${\cal C}(\Gamma, V)$.   On account of the uniqueness of ${\cal C}(\Gamma, V)$ assigned to a given pair $(\Gamma, V)$, one can also 
regard that our criterion refers to separability of states with respect to ${\cal C}(\Gamma, V)$.  This may be more natural from the physical viewpoint that correlations are evaluated only with respect to the set of observables under consideration/preparation in advance.   Now, the aforementioned fact motivating the definition is shown by
\begin{theorem}
A normalized $N$-particle antisymmetric pure state $\ket{\Psi}  \in {\cal H}_A^{\rm tot}$ is separable with respect to $(\Gamma, V)$ if and only if  $\ket{\Psi} \in S(\Gamma, V)$. 
\end{theorem}
\begin{proof}
Suppose $\ket{\Psi}\in S(\Gamma, V)$.  Applying (\ref{decO}) of Corollary 1 to $P^{(k)} = P^{(k)}(\Gamma, V)$ in (\ref{defPk}), we obtain
\be
\bra{\Psi}P^{(k)}\ket{\Psi}=\bra{\psi_k}P_k\ket{\psi_k}.
\label{Pdef}
\ee
Setting 
$P_k=\ket{\psi_k}\bra{\psi_k}$ for all $k$, we see
\be
\bra{\Psi}P^{(k)}\ket{\Psi}=1, \qquad k = 1, \ldots, s.
\ee

Conversely, suppose that there exists a state $\ket{\Phi}\in{\cal H}_A^{\rm tot}$ fulfilling  (\ref{defsep}).  We then have
\be
P^{(k)}\ket{\Phi}=\ket{\Phi}+\ket{\Omega_k}
\label{PhiOmega}
\ee
for some state $\ket{\Omega_k}\in{\cal H}_A^{\rm tot}$ with
\be
\braket{\Phi|\Omega_k}=0.
\label{POortho}
\ee
Since $P^{(k)}$ is a projector, using (\ref{PhiOmega}) and (\ref{POortho}),  we obtain
\be
1=\bra{\Phi}P^{(k)}\ket{\Phi}
=\bra{\Phi}(P^{(k)})^2\ket{\Phi}
=\(\bra{\Phi}+\bra{\Omega_k}\)\(\ket{\Phi}+\ket{\Omega_k}\)
=1+\braket{\Omega_k|\Omega_k},
\ee
from which we find $\ket{\Omega_k}=0$.  This implies
$P^{(k)}\ket{\Phi}=\ket{\Phi}$ for all $k$, and thus if we consider the product of the projectors,
\be
\label{pprod}
P := \prod_{k=1}^s P^{(k)} =A\(\Bot_{k=1}^sP_k\)A,
\ee
we find
\be
P\,\ket{\Phi} =  \prod_{k=1}^sP^{(k)}\ket{\Phi}=\ket{\Phi}.
\label{multiphi}
\ee
Note that since all $P^{(k)}$ belong to ${\cal C}(\Gamma, V)$ defined in
(\ref{clde}), so does the product $P \in {\cal C}(\Gamma, V)$ by Corollary 2.
On the other hand, since the operator $P_k$ in $P^{(k)}$ is a
one-dimensional projector in $W(\Gamma_k,V_k)$, one can put $P_k =
\ket{\xi_k}\bra{\xi_k}$ for some normalized $\ket{\xi_k}\in
W(\Gamma_k,V_k)$.   Then, recalling (\ref{pprod}), we have
 \be
P =
A\left(\bigotimes_{k=1}^s\ket{\xi_k}\bra{\xi_k}\right)A
=\ket{\Xi}\bra{\Xi},
 \label{phisep}
\ee
where
 \be
 \ket{\Xi}:=A \bigotimes_{k=1}^s\ket{\xi_k} \in S(\Gamma, V).
 \label{PinS}
\ee
Combining (\ref{phisep}) with (\ref{multiphi}), we arrive at
\be
\ket{\Phi}=\braket{\Xi|\Phi}\ket{\Xi}.
\label{PP}
\ee
which shows $\ket{\Phi}\in S(\Gamma, V)$.
\end{proof}

\medskip

We learn from (\ref{PP}) that $\ket{\Phi}$ is equivalent to $\ket{\Xi}$ in (\ref{PinS}) up to the global phase $\braket{\Xi|\Phi}$.  It then follows that there exists a one-to-one correspondence between a projector $P$ and a state in $S(\Gamma, V)$ , and if $\ket{\Phi}$ and $P$ are in such a correspondence, we have $\bra{\Phi} P \ket{\Phi}= 1$.   
Concerning the ontological aspect of our separability criterion, Theorem 2 implies that separable states
under the partition $\Gamma$ consists of states $\ket{\xi_k} \in W(\Gamma_k,V_k)$, $k = 1, \ldots, s$, each possessing 
the \lq complete set of properties\rq\ in the subsystem $\Gamma_k$ as mentioned in \cite{GMW02}.   
In fact, the measurement outcome $\bra{\Phi} P \ket{\Phi} = 1$ (or (\ref{defsep})) ensures that the possession of the properties specified by $P_k$ can be confirmed with certainty to the set of number $\vert \Gamma_k \vert$ of particles, not to the individual particles in each of the subsystems $\Gamma_k$, in accordance with the indistinguishability of the particles.  Note that this holds simultaneously for all $k$, since $P^{(k)}(\Gamma, V)$ commute with each other among themselves.

We point out that it is possible to regard our criterion (Definition 1) as a generalization of the criterion given in \cite{GMW02} for the bipartition $s = 2$ case in which only one projection (rather than $s=2$ projections) is used.   Indeed, we see that for the arbitrary $s$-partite case our criterion can actually be replaced by a slightly more economical one which uses $(s-1)$ projections, that is,
\begin{proposition}
A normalized $N$-particle antisymmetric pure state $\ket{\Psi}\in {\cal H}_A^{\rm tot}$ is separable with respect to $(\Gamma,V)$ if and only if there exist $(s-1)$ one-dimensional projection operators $P_k$  in $W(\Gamma_k,V_k)$ from ${\cal H}_A(\Gamma_k)$ for $k=1,\cdots, t-1, t+1,\cdots,s$ with an integer $t \in [1, s]$, such that
(\ref{defsep}) are fulfilled with (\ref{defPk}).
\end{proposition}
\begin{proof}
Consider the product of $(s-1)$ projectors,
\be
 P' :=  \(\prod_{k=1}^{t-1} P^{(k)}\) \(\prod_{k=t+1}^{s} P^{(k)} \) = A\(\(\bigotimes_{k=1}^{t-1}P_k\)\otimes \mathbb{I}_t\otimes\(\bigotimes_{k=t+1}^s P_k\)\)A.
\ee
$P'$ is a projection operator because, from Lemma 4, we have
\be
 P'^2 = A\(\(\bigotimes_{k=1}^{t-1}P_k^2\)\otimes \mathbb{I}_t^2\otimes\(\bigotimes_{k=t+1}^s P_k^2\)\)A
  =  A\(\(\bigotimes_{k=1}^{t-1}P_k\)\otimes \mathbb{I}_t\otimes\(\bigotimes_{k=t+1}^s P_k\)\)A
  = P'.
\ee
The same argument used to show (\ref{multiphi}) then implies that, if a state $\ket{\Psi}\in {\cal H}_A^{\rm tot}$ fulfills (\ref{defsep}) for $k \ne t$, we have
\be
 P'\ket{\Psi} = \ket{\Psi}.
\ee
Since $P_k$ is a one-dimensional projection operator in $W(\Gamma_k,V_k)$,
one can put $P_k=\ket{\xi_k}\bra{\xi_k}$ for some normalized $\ket{\xi_k}\in W(\Gamma_k,V_k)$.
Using any normalized basis $\{\ket{\mu_t}\}$ in $W(\Gamma_t,V_t)$, we also have 
$ \mathbb{I}_{t}= \sum_{\mu_t}\ket{\mu_t}\bra{\mu_t}$.
Then $P'$ can be written as
\be
 P' = A\(\(\bigotimes_{k=1}^{t-1}\ket{\xi_k}\bra{\xi_k}\)\otimes \(\sum_{\mu_t}\ket{\mu_t}\bra{\mu_t}\)\otimes\(\bigotimes_{k=t+1}^s \ket{\xi_k}\bra{\xi_k}\)\)A
  = \sum_{\mu_t}\ket{\Xi_{\mu_t}}\bra{\Xi_{\mu_t}},
\ee
where
\be
 \ket{\Xi_{\mu_t}}= A\( \(\bigotimes_{k=1}^{t-1}\ket{\xi_k}\)\otimes \ket{\mu_t}\otimes\(\bigotimes_{k=t+1}^s \ket{\xi_k}\)\).
\ee
It follows that
\be
\begin{split}
 \ket{\Psi} 
 &= P'\ket{\Psi}\\
 &= \sum_{\mu_t}\braket{\Xi_{\mu_t}|\Psi}\ket{\Xi_{\mu_t}}\\
 &= A\( \(\bigotimes_{k=1}^{t-1}\ket{\xi_k}\)\otimes\(\sum_{\mu_t}\braket{\Xi_{\mu_t}|\Psi} \ket{\mu_t}\)\otimes\(\bigotimes_{k=t+1}^s \ket{\xi_k}\)\)\in S(\Gamma,V).
\end{split}
\ee
The converse is trivial.
\end{proof}

The criterion for the full separation, where we have $s = N$ and $\Gamma_k = \{k\}$ for $k = 1, 2, \ldots, N$,  is also given in \cite{GMW02} 
with $N$ projectors (see (\ref{epcond})  below) in a way slightly different from the bipartition $s=2$ case.   
However, again, the equivalence to our criterion in that case can be established by

\begin{proposition}
An $N$-particle antisymmetric pure state $\ket{\Psi} \in \mathcal{H}_A^{\rm tot}$ is fully separable if and only if there exist one-dimensional projection operators $P_i$, $i=1,2, \cdots,N$, which are commonly defined in each constituent Hilbert space $\mathbb{C}^d$ and are mutually orthogonal  $P_iP_j=0$ for $i\neq j$ such that
\be
 \label{epcond}
 \bra{\Psi}E^{(i)} \ket{\Psi} = 1,
\ee
with
\be
 E^{(i)} = \mathbbm{1}^{\ot N} - \(\mathbbm{1}-P_i\)^{\ot N}.
\ee
\end{proposition}
\begin{proof}
We first expand $E^{(i)}$ as 
\be
 E^{(i)} 
 =-\sum_{\sigma\in S_N}\sum_{k=1}^N\frac{(-1)^{k}}{N!}{N\choose k}\pi_\sigma\(P_i^{\ot k}\ot\mathbbm{1}^{\ot(N-k)}\)\pi_\sigma,
\ee
and thereby consider their product $\prod_{i=1}^N E^{(i)}$.  
Since each term in $E^{(i)}$ has at least one projector $P_i$ in the constituents, and since
they are mutually orthogonal $P_iP_j=0$ for $i\neq j$, we find
\be
E := \prod_{i=1}^N E^{(i)} = \sum_{\sigma\in S_N}\pi_\sigma\(\Bot_{i=1}^N P_i\) \pi_\sigma.
\ee
On account of (\ref{Api}), we obtain
\be
 \label{eaident}
 \mathcal{A} E \mathcal{A} 
  = \sum_{\sigma\in S_N}\mathcal{A}\pi_\sigma\(\Bot_{i=1}^N P_i\) \pi_\sigma \mathcal{A}
  = \sum_{\sigma\in S_N}\mathcal{A}\(\Bot_{i=1}^N P_i\) \mathcal{A}
  = N!\, \mathcal{A}\(\Bot_{i=1}^N P_i\) \mathcal{A} = P,
\ee
where $P$ is the product defined in (\ref{pprod}) whose multinomial coefficient (\ref{cardm}) is now $M(\Gamma) = N!$.

Suppose that a state $\ket{\Psi}\in {\cal H}_A^{\rm tot}$ fulfills (\ref{epcond}) for all $i$.  Since $E^{(i)}$ are all projectors,
employing again the same argument to reach (\ref{multiphi}) we find
\be
 E\ket{\Psi} = \ket{\Psi}.
\ee
It then follows from (\ref{eaident}) and $\mathcal{A} \ket{\Psi} = \ket{\Psi}$ that
\be
P \ket{\Psi} = \mathcal{A} E \mathcal{A} \ket{\Psi} = \mathcal{A} E \ket{\Psi} = \ket{\Psi}.
\ee
This is just the condition (\ref{multiphi}), and hence we arrive at the same conclusion $\ket{\Psi}\in S(\Gamma, V)$ as before.   The converse is trivial.
\end{proof}

\begin{figure}
\includegraphics[height=1.4in]{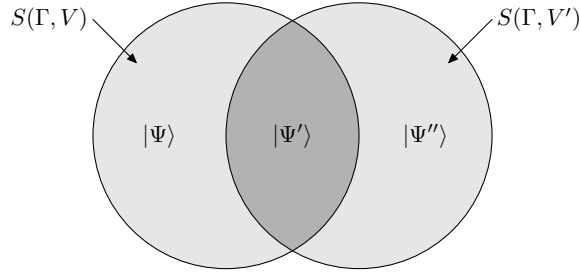}\\
\caption{%
A state may be separable with respect to more than one orthogonal structure under an identical partition $\Gamma$ and, as a result, the property of separability may not be transitive.   Here, both of the states $\ket{\Psi}$ and $\ket{\Psi^\prime}$ belong to $S(\Gamma, V)$ and hence they are mutually separable ({\it e.g.}, their inner product is factorizable).  Similarly,  
$\ket{\Psi^\prime}$ and $\ket{\Psi^{\prime\prime}}$ belong to $S(\Gamma, V')$, and they are also mutually separable.  
However, the two states $\ket{\Psi}$ and $\ket{\Psi^{\prime\prime}}$ are not necessarily separable, 
unless they share a common orthogonal structure.
}
\label{figone}
\end{figure}

A notable feature of our criterion is that it renders explicit the dependence of separability 
on the orthogonal structure $V$ used.    An important consequence of this is that such a structure may not be assigned to a state uniquely, that is, 
a single state $\ket{\Psi}$ may admit more than two orthogonal structures under a given partition $\Gamma$.
Moreover, the very separable nature of states is not necessarily transitive, that is,
even if $\ket{\Psi}$ and $\ket{\Psi^\prime}$, and also $\ket{\Psi^\prime}$ and $\ket{\Psi^{\prime\prime}}$, are separable between them sharing the same orthogonal structures, respectively, $\ket{\Psi}$ and $\ket{\Psi^{\prime\prime}}$ may not be separable, unless they share the same orthogonal structure (see Figure \ref{figone}).
These properties will be seen more explicitly by examples in the next section.  

Prior to this, however, we add
\begin{definition}
{\it An $N$-particle antisymmetric pure state $\ket{\Psi}\in{\cal H}_A^{\rm tot}$ is {\rm separable} if it is separable
with respect to $(\Gamma, V)$ for some partition $\Gamma$ and set of orthogonal spaces $V$.}
\end{definition}
This form of criterion will be more convenient when one is working in situations where the dependence of $V$ is immaterial and only the possibility of separability of states is of concern.

\section{Implications of the orthogonal structures: examples}
\setcounter{equation}{0}  

In this section, related to the orthogonal structures encountered in our criterion, we elaborate on two technical points 
which require particular attention in the discussion of separability.  These are non-uniqueness of the space $S(\Gamma, V)$ that can be assigned to a factorizable state,  and non-transitivity of factorizability among individual factorizable states.   We also touch upon the physical role of the orthogonal structures concerning the cluster decomposition property.   For definiteness, in this section we illustrate  these points by examples using only normalized states both for the total system and the subsystems.

\subsection{Non-uniqueness}

In Sect.~III we mentioned that a state may belong to $S(\Gamma, V)$ for more than one $V$.   To see this,  
consider a three particle $N=3$ state with $d=4$,
\be
\ket{\Psi}=A\[\frac{1}{\sqrt{2}}\(\ket{0}_1\ket{1}_2-\ket{1}_1\ket{0}_2\)\, \ket{3}_3\],
\ee
where we employ abbreviated notations such as  $\ket{0}_1\ket{1}_2$ in place of $\ket{0}_1\ot\ket{1}_2$. Clearly,  the state $\ket{\Psi}$ belongs to $S(\Gamma, V)$ with
 \be
 \label{gmone}
 \Gamma=\{\Gamma_1, \Gamma_2\},
 \qquad
 \Gamma_1=\{1,2\},
 \qquad
 \Gamma_2=\{3\}
 \ee
and
\be
 V=\{V_1, V_2\},
 \qquad
 V_1={\rm span}\{\ket{0}, \ket{1}\},
 \qquad
 V_2={\rm span}\{\ket{2}, \ket{3}\}.
\ee
On the other hand, $\ket{\Psi}$ also belongs to $S(\Gamma, V^\prime)$ with the same $\Gamma$ in (\ref{gmone}) and
\be
 V^\prime=\{V^\prime_1, V^\prime_2\},
 \qquad
 V_1^\prime={\rm span}\{\ket{0}, \ket{1}, \ket{2}\},
 \qquad
 V_2^\prime={\rm span}\{\ket{3}\}.
\ee
Consequently, the state $\ket{\Psi}$ is separable with respect to both $(\Gamma, V)$ and $(\Gamma, V^\prime)$.

Note that in the above example the two orthogonal structures $V$ and $V'$ are considered under the same basis vectors $\{\ket{a}_i\}$.  
However, this is not necessary as can be seen in the familiar case of the singlet state,
 \be
 \begin{split}
 \ket{\Psi}&=A\(\ket{0}_1\ket{1}_2\)=\frac{1}{\sqrt{2}}\(\ket{0}_1\ket{1}_2-\ket{1}_1\ket{0}_2\)
 \end{split}
 \ee
in a two particle $N=2$ system with $d=2$.   Indeed, we see immediately that $ \ket{\Psi} \in S(\Gamma, V)$ for
  \be
 \Gamma=\{\Gamma_1, \Gamma_2\},
 \qquad
 \Gamma_1=\{1\},
 \qquad
 \Gamma_2=\{2\}
 \label{partex}
 \ee
and infinitely many sets of spaces $V(\theta,\phi)=\{V_1(\theta,\phi), V_2(\theta,\phi)\}$ with
 \be
V_1(\theta,\phi)=
 {\rm span}\{
 \cos\theta\ket{0}+\e^{\im\phi}\sin\theta\ket{1}\},
  \qquad
  V_2(\theta,\phi)=
{\rm span}\{
  \sin\theta\ket{0}-\e^{-\im\phi}\cos\theta\ket{1}
\}
\ee
parametrized by angles  $0\le\theta\le \pi$ and $0\le\phi<2\pi$.  Thus, the singlet state $\ket{\Psi}$ is separable with respect to all such $(\Gamma, V(\theta,\phi))$.
However, this ambiguity is deceptive because the singlet state $\ket{\Psi}$ is the only element spanning the total Hilbert space ${\cal H}_A^{\rm tot}$ in this case, and hence the resultant space $S(\Gamma, V(\theta,\phi))$ is actually unique.   Accordingly, all observables in ${\cal C}(\Gamma, (\theta,\phi))$, or even those in ${\cal C}(\Gamma)$, are proportional to the one dimensional projector $\ket{\Psi}\bra{\Psi}$.  For example, from (\ref{clde}) one finds that any observable in ${\cal C}(\Gamma, V(0,0))$ is written as
 \be
 O=A(k_1\ket{0}_1{_1}\bra{0}\otimes k_2\ket{1}_2{_2}\bra{1})A=k_1k_2\ket{\Psi}\bra{\Psi},
 \qquad
 k_1,k_2\in\R,
 \label{Osinglet}
 \ee
 which is obviously the case for all ${\cal C}(\Gamma,V(\theta,\phi))$.

At this point, we remark on the outcome of joint spin measurements of two particles obtained by the
spin operator $\sigma(\vect{a}, \vect{b}) := \sigma_1(\vect{a})\ot\sigma_2(\vect{b})$ where $\sigma_1(\vect{a})=\vect{a}\cdot {\bf \sigma}_1$ and $\sigma_2(\vect{b})=\vect{b}\cdot \sigma_2$ are
defined from normalized real $3$-vectors $\vect{a}$, $\vect{b}$  and Pauli matrices ${\sigma}_i = (\sigma_i^x, \sigma_i^y, \sigma_i^z)$ for $i = 1, 2$.
Since $\sigma(\vect{a}, \vect{b})$ acts on 
${\cal H}^{\rm tot} = {\cal H}_1\otimes {\cal H}_2$ but not on ${\cal H}_A^{\rm tot}$, for the fermionic system we may consider the corresponding operator  ${\cal A}\,\sigma(\vect{a} , \vect{b})\,{\cal A}$ which is in the form (\ref{defO}) and belongs to ${\cal C}(\Gamma)$.
However, it is important to notice
that, if we are to regard $O_1 = {\bf \sigma}_1(\vect{a})$ and $O_2 = \sigma_2(\vect{b})$ as a set of operators for individual spin measurements of the particles, then $\sigma(\vect{a}, \vect{b})$ does {\it not}  belong to ${\cal C}(\Gamma, V)$, because the spin operators, when considered in an identical subsystem $\Gamma_i$, have a non-vanishing support
${\bf \sigma}_i(\vect{a}){\bf \sigma}_i(\vect{b})\neq0$ for $i=1,2$ in general and are not orthogonal to each other (cf.~(\ref{ospecw})).   Moreover, one can see from
\be
{\cal A}\,\sigma(\vect{a}, \vect{b})\, {\cal A}= -(\vect{a}\cdot\vect{b}) \ket{\Psi}\bra{\Psi}
\label{asaop}
\ee
that it is impossible to find an orthogonal set of operators $O_1$ and $O_2$ for ${\cal A}\,\sigma(\vect{a} , \vect{b})\,{\cal A}$ to form an observable $O \in {\cal C}(\Gamma,V)$ as (\ref{clde}) such that they depend on each of the measurement parameters $\vect{a}$ and $\vect{b}$ separately, $O_1 = O_1(\vect{a})$, $O_2 = O_2(\vect{b})$, 
because if so the resultant operator would then be of the form (\ref{Osinglet}) with the factorized eigenvalues $k_1 = k_1(\vect{a})$ and $k_2 = k_2(\vect{b})$ in contradiction with (\ref{asaop}).
This indicates that, if we restrict ourselves only to the spin degrees of freedom, 
the formal non-factorizability of the correlation $\bra{\Psi}\sigma(\vect{a},\vect{b})\ket{\Psi}=-\vect{a}\cdot\vect{b}$ (from which  the Cirelson operator is formed \cite{CHSH, Cirelson80}) for the singlet state \cite{Bellone}, despite it belongs to the separable set $S(\Gamma, V)$, can 
arise from the fact that the operator  ${\cal A}\,\sigma(\vect{a} , \vect{b})\,{\cal A}$ does not belong to ${\cal C}(\Gamma,V)$ under the identification $O_1 = {\bf \sigma}_1(\vect{a}),  \, O_2 = \sigma_2(\vect{b})$ based on which the spin correlation is considered.  
A clearer account of spin correlation for fermions, with the orthogonal structure $V$ equipped naturally to the actual experimental setup, can be gained by introducing the spatial degrees of freedom possessed by the particles, as we shall see in the last examples in this section.

To sum up, the lesson one learns here is that, even if a state is separable with respect to some $(\Gamma, V)$, the expectation value may not be separable for observables unless they belong to the corresponding class ${\cal C}(\Gamma, V)$ in which individual measurements of subsystems can be defined properly.  

\subsection{Non-transitivity}

In Sect.~III we noted that, if two states $\ket{\Psi}$ and $\ket{\Psi^\prime}$ belong to one $S(\Gamma, V) \subseteq {\cal H}_A^{\rm tot}$,  their inner product is factorizable.   It is important to notice that this property is not transitive, that is, even if the 
inner product is factorizable between $\ket{\Psi}$ and $\ket{\Psi^\prime}$, and also between $\ket{\Psi^\prime}$ and $\ket{\Psi^{\prime\prime}}$, it does not necessarily imply that it is so between  $\ket{\Psi}$ and $\ket{\Psi^{\prime\prime}}$, unless all the three belong to the same $S(\Gamma, V)$ (see Figure \ref{figone}).
To illustrate this, let us consider an $N=2$ system with $d=6$.  We have $\Gamma$ in (\ref{partex}) and consider the two sets of space,  $V=\{V_1, V_2\}$ with
\be
V_1={\rm span}\{\ket{a}\,|\, a=0,1,2\},
 \qquad
 V_2={\rm span}\{\ket{a}\,|\, a=3,4,5\},
\ee
and $V^\prime=\{V^\prime_1, V^\prime_2\}$ with
\be
V_1^\prime={\rm span}\{\ket{a}\,|\, a=0,1,3\},
 \qquad
 V_2^\prime={\rm span}\{\ket{a}\,|\, a=2,4,5\}.
\ee
In terms of normalized vectors $\alpha=(\alpha_1,\alpha_2,\alpha_3)\in\C^3$, the state vectors in $V_i[{\cal H}_i]$ and $V_i^\prime[{\cal H}_i]$ can be written as
\be
\begin{split}
\ket{\psi_1(\alpha)} =\alpha_{1}\ket{0}_1+\alpha_{2}\ket{1}_1+\alpha_{3}\ket{2}_1\in V_1[{\cal H}_1],
\qquad
\ket{\psi_2(\alpha)} =\alpha_{1}\ket{3}_2+\alpha_{2}\ket{4}_2+\alpha_{3}\ket{5}_2\in V_2[{\cal H}_2],
\end{split}
\ee
and
\be
\begin{split}
\ket{\psi^\prime_1(\alpha)} =\alpha_{1}\ket{0}_1+\alpha_{2}\ket{1}_1+\alpha_{3}\ket{3}_1\in V_1^\prime[{\cal H}_1],
\qquad
\ket{\psi^\prime_2(\alpha)} =\alpha_{1}\ket{2}_2+\alpha_{2}\ket{4}_2+\alpha_{3}\ket{5}_2\in V_2^\prime[{\cal H}_2],
\end{split}
\label{Deltavec}
\ee
respectively.
The states in $S(\Gamma, V)$ and $S(\Gamma, V^\prime)$ can thus be obtained as
\be
\ket{\Psi(\alpha, \alpha^\prime)} = A\ket{\psi_1(\alpha)}\ket{\psi_2(\alpha^\prime)}\in S(\Gamma, V),
 \qquad
\ket{\Psi^\prime(\alpha, \alpha^\prime)} = A\ket{\psi^\prime_1(\alpha)}\ket{\psi^\prime_2(\alpha^\prime)}\in S(\Gamma, V^\prime).
\ee
Observe that, although $S(\Gamma, V)$ and $S(\Gamma, V^\prime)$ are different, they are not exclusive.  Indeed, choosing
\be
\beta=(\beta_{1},\beta_{2},0)
\qquad
{\rm and}
\qquad
\beta^\prime=(0, \beta_{2}^\prime,\beta_{3}^\prime),
\ee
we find
\be
\ket{\Psi(\beta,\beta^\prime)}=\ket{\Psi^\prime(\beta,\beta^\prime)}\in S(\Gamma, V)\cap S(\Gamma, V^\prime).
\label{cap}
\ee

Let now $\bar S(\Gamma, V)$ be the complement of $S(\Gamma, V)$ in ${\cal H}_A^{\rm tot}$ satisfying
\be
S(\Gamma, V)\cap\bar S(\Gamma, V)=\emptyset
\qquad
{\rm and}
\qquad
S(\Gamma, V)\cup\bar S(\Gamma, V)={\cal H}_A^{\rm tot}.
\ee
If we choose two states $\ket{\Psi(\alpha, \alpha^\prime)}$ and $\ket{\Psi^\prime(\gamma, \gamma^\prime)}$ such that 
\be
\ket{\Psi(\alpha, \alpha^\prime)}\in S(\Gamma, V)\cap\bar S(\Gamma, V^\prime)
\qquad
{\rm and}
\qquad
\ket{\Psi^\prime(\gamma, \gamma^\prime)}\in\bar S(\Gamma, V)\cap S(\Gamma, V^\prime),
\label{defAC}
\ee
then from Lemma 2 we find that their inner products with $\ket{\Psi(\beta,\beta^\prime)}$ in  (\ref{cap}) are factorizable,
\be
\begin{split}
\braket{\Psi(\alpha, \alpha^\prime)|\Psi(\beta, \beta^\prime)}
&=\braket{\psi_1(\alpha)|\psi_1(\beta)}\braket{\psi_2(\alpha^\prime)|\psi_2(\beta^\prime)},\\
\braket{\Psi^\prime(\beta, \beta^\prime)|\Psi^\prime(\gamma, \gamma^\prime)}
&=\braket{\psi^\prime_1(\beta)|\psi^\prime_1(\gamma)}\braket{\psi^\prime_2(\beta^\prime)|\psi^\prime_2(\gamma^\prime)}.
\end{split}
\label{sepAB}
\ee
However, it is readily confirmed that 
\be
\braket{\Psi(\alpha, \alpha^\prime)|\Psi^\prime(\gamma, \gamma^\prime)}\neq\braket{\psi_1(\alpha)|\psi^\prime_1(\gamma)}\braket{\psi_2(\alpha^\prime)|\psi_2^\prime(\gamma^\prime)}
\label{insep}
\ee
in general, showing that the factorizability of the inner product is not transitive in this case.

\subsection{Cluster Decomposition Property} 

Consider an $N=2$ system with $d = 4$.  We regard the constituent Hilbert space $\mathbb{C}^4$ as $\mathbb{C}^2\otimes\mathbb{C}^2$, where the first $\mathbb{C}^2$ is spanned by the basis $\{\ket{+},\ket{-}\}$ while the second $\mathbb{C}^2$ is spanned by 
$\{\ket{L}, \ket{R}\}$.   The names of the basis vectors are motivated by the idea that set $\{\ket{+},\ket{-}\}$ refers to internal spin states,
and the set $\{\ket{L}, \ket{R}\}$ refers to localized states on the left/right in space which are separated from each other; for example, $\ket{+}\ket{L}$ implies that the particle is localized on the left with the spin state $\ket{+}$.  

As before, our $\Gamma$ is given in (\ref{partex}), and if we 
choose the set $V$ as
 \be
 V=\{V_1,V_2\},
 \qquad
 V_1={\rm span}\{\ket{a}\ket{L}\,|\, a=+,-\},
 \qquad
 V_2={\rm span}\{\ket{a}\ket{R}\,|\, a=+,-\},
 \label{Vex3}
\ee
then we  can write the states in $S(\Gamma, V)$ as
\be
 \ket{\Psi(\alpha,\alpha^\prime)}
=A\(\ket{\alpha}_1\ket{L}_1 \ket{\alpha^\prime}_2\ket{R}_2\)
=\frac{1}{\sqrt{2}}\(\ket{\alpha}_1\ket{L}_1 \ket{\alpha^\prime}_2\ket{R}_2
 -\ket{\alpha^\prime}_1\ket{R}_1 \ket{\alpha}_2\ket{L}_2\),
 \label{AStwo}
 \ee
 where $\ket{\alpha}_i$ and $\ket{\alpha^\prime}_i$, $i=1,2$, denote the spin states of the $i$-th particle.
 Namely, the space $S(\Gamma, V)$ consists of states of two particles one of which is on the left and the other on the right, and
 the vanishing overlap $\braket{L|R}=0$ expresses the remoteness between them. 
 
 With the observables $\sigma_1(\vect{a})$
 and $\sigma_2(\vect{b})$ in the spin space of the two particles defined earlier, respectively, we consider the observables
 \be
 O_1 = \sigma_1(\vect{a})\, \ket{L}_1{_1}\bra{L},
 \qquad
  O_2 = \sigma_2(\vect{b})\,  \ket{R}_2{_2}\bra{R}.
 \ee
We then construct the observable $O\in{\cal C}(\Gamma, V)$ on the two particles according to the prescription (\ref{clde}):
\be
O=A\( O_1 \otimes O_2 \)A
=A\(\sigma(\vect{a},\vect{b})\ot\ket{L}_1{_1}\bra{L}\otimes\ket{R}_2{_2}\bra{R}\)A.
\label{symob}
\ee
This corresponds clearly to the simultaneous measurement of spins on one particle on the left and on the other on the right, and for
the state $ \ket{\Psi(\alpha,\alpha^\prime)} \in S(\Gamma, V)$ in (\ref{AStwo}) its expectation value factorizes
\be
\bra{\Psi(\alpha,\alpha^\prime)}O\ket{\Psi(\alpha,\alpha^\prime)}
={_1}\bra{\alpha} \sigma_1(\vect{a}) \ket{\alpha}_1\, {_2}\bra{\alpha^\prime} \sigma_2({\bf b}) \ket{\alpha^\prime}_2
={_1}\bra{\alpha^\prime} \sigma_1(\vect{b})  \ket{\alpha^\prime}_1\, {_2}\bra{\alpha}\sigma_2(\vect{a}) \ket{\alpha}_2,
\label{CDP}
\ee
by Corollary 2.  Note that, on account of Lemma 1, the factorization is also seen for the operator $\tilde O$ in (\ref{defcalO}),
\be
\tilde O= \sigma(\vect{a},\vect{b})\ot \ket{L}_1{_1}\bra{L}\ot \ket{R}_2{_2}\bra{R} + 
\sigma(\vect{b},\vect{a})\ot\ket{R}_1{_1}\bra{R}\ot\ket{L}_2{_2}\bra{L}.
\label{symoc}
\ee
The outcome shows that, even though the states (\ref{AStwo}) are not direct product states because of the (rescaled) anti-symmetrizer $A$, 
the correlation of the two local measurements $O_1$ and $O_2$ disappears.  In other words, if the two fermions are localized and remotely separated from each other, they behave independently 
as long as their spins are concerned.
This is a prime example to illustrate the cluster decomposition property \cite{WC63, Peres93, Weinberg95} in the fermionic system, and has been mentioned in \cite{GMW02}.

However, the factorizability may break down for states which admit intrinsically long range correlations.
To see this, let us consider the antisymmetric state,
\be
\begin{split}
\ket{\Phi(\alpha,\alpha^\prime)}
=A\(\ket{\psi^+(\alpha,\alpha^\prime)}\otimes\ket{L}_1\ket{R}_2\)
=\ket{\psi^+(\alpha,\alpha^\prime)}\otimes\frac{1}{\sqrt2}\(\ket{L}_1\ket{R}_2-\ket{R}_1\ket{L}_2\),
\end{split}
\ee
where
\be
\label{alplus}
\ket{\psi^+(\alpha,\alpha^\prime)}=N(\alpha,\alpha^\prime)\(\ket{\alpha}_1\ket{\alpha^\prime}_2+\ket{\alpha^\prime}_1\ket{\alpha}_2\),
\qquad
N(\alpha,\alpha^\prime)=\frac{1}{\sqrt{2(1+|{}_1\!\braket{\alpha|\alpha^\prime}\!{}_1|^2)}},
\ee
where we used ${}_1\braket{\alpha^\prime|\alpha}{}_1={}_2\braket{\alpha^\prime|\alpha}{}_2$ derived from the isomorphism between the constituents.
Clearly, since the state $\ket{\Phi(\alpha,\alpha^\prime)}$ is not an anti-symmetrized state of a direct product of {\it constituent} states,  we have $\ket{\Phi(\alpha,\alpha^\prime)}\notin S(\Gamma, V)$ for $S(\Gamma, V)$ given in (\ref{Vex3}). 
Such a state with orthogonality in the internal degree of freedom $\braket{\alpha|\alpha^\prime}=0$ can be found as a 1s2s state of electrons in Helium atom in the first order approximation of  perturbation theory, where 
the state for the spatial degrees of freedom of $\ket{\Phi(\alpha,\alpha^\prime)}$ results from the anti-symmetrization of the product state $\ket{L}_1\ket{R}_2$.   Because of the factorization between the spin part and the spatial part, we have 
the density matrix for the spin part in the pure state form,
\be
\sum_{i,j=L,R}{_1}\bra{i}{_2}\bra{j}\(\ket{\Phi(\alpha,\alpha^\prime)}\bra{\Phi(\alpha,\alpha^\prime)}\)\ket{i}_1\ket{j}_2=\ket{\psi^+(\alpha,\alpha^\prime)}\bra{\psi^+(\alpha,\alpha^\prime)}.
\ee
The expectation value of the observable $O$ in (\ref{symob}) with $\ket{\Psi(\alpha,\alpha^\prime)}$ then becomes
\be
\begin{split}
\bra{\Phi(\alpha,\alpha^\prime)}O\ket{\Phi(\alpha,\alpha^\prime)}
=N(\alpha,\alpha^\prime)^2&\Bigl({_1}\bra{\alpha}\sigma_1(\vect{a}) \ket{\alpha}_1\, {_2}\bra{\alpha^\prime} \sigma_2(\vect{b})\ket{\alpha^\prime}_2
+{_1}\bra{\alpha^\prime} \sigma_1(\vect{a} \ket{\alpha^\prime}_1\,{_2}\bra{\alpha}\sigma_2(\vect{b}) \ket{\alpha}_2\\
&+{_1}\bra{\alpha}\sigma_1(\vect{a}) \ket{\alpha^\prime}_1\, {_2}\bra{\alpha^\prime} \sigma_2(\vect{b})\ket{\alpha}_2
+{_1}\bra{\alpha^\prime} \sigma_1(\vect{a}) \ket{\alpha}_1\,{_2}\bra{\alpha}\sigma_2(\vect{b}) \ket{\alpha^\prime}_2\Bigr),
\end{split}
\ee
which is no longer factorized and gives rise to correlations between the constituents in general. Also, utilizing the relations ${_1}\bra{\alpha}\sigma_1(\vect{a}) \ket{\alpha}_1={_2}\bra{\alpha}\sigma_2(\vect{a}) \ket{\alpha}_2$ and so on which come from the isomorphism between constituents, one can confirm Proposition 2 by direct calculation of the expectation value of $\tilde{O}$ in (\ref{symoc}) with respect to $\ket{\Phi(\alpha,\alpha^\prime)}$.

To extend the above examples to a system with more than three subsystems, let us consider an $N = 4$ system with $d=6$. The constituent Hilbert space $\mathbb{C}^6$ can be thought of as the direct product of $\mathbb{C}^2$ spanned by the basis $\{\ket{+},\ket{-}\}$ and $\mathbb{C}^3$ spanned by the basis $\{\ket{L},\ket{C},\ket{R}\}$.  As before, 
we regard $\{\ket{+},\ket{-}\}$ as internal spin states and $\{\ket{L},\ket{C},\ket{R}\}$ as spatial states localized on the left, the center and the right, respectively, which are separated from each other and share no support in common.   We choose the partition $\Gamma$ and the set of subspaces $V$ as
\be
\Gamma = \{\Gamma_1,\Gamma_2,\Gamma_3\},
\qquad
\Gamma_1 = \{1\},
\qquad
\Gamma_2 =\{2\},
\qquad
\Gamma_3=\{3,4\},
\label{quatro}
\ee
and
$V=\{V_1,V_2,V_3\}$
with
\be
V_1=\mathrm{span}\{\ket{a}\ket{L}|a=+,-\},
\qquad 
V_2=\mathrm{span}\{\ket{a}\ket{C}|a=+,-\},
\qquad
V_3=\mathrm{span}\{\ket{a}\ket{R}|a=+,-\}.
\ee
The states in $S(\Gamma,V)$ can be written as
\be
\ket{\Psi(\alpha,\beta,\gamma)} = A\(\ket{\alpha}_1\ket{L}_1\ket{\beta}_2\ket{C}_2\ket{\gamma}_{34}\),
\ee
where 
\be
\ket{\gamma}_{34} = 
\ket{{\rm Singlet}}_{34}\ot\ket{R}_3\ket{R}_4\in W(\Gamma_3,V_3)
\ee
is the state of subsystem $\Gamma_3$ representing the spin singlet state localized in the right of the space.  From the prescription in (\ref{clde}), the observables in ${\cal C}(\Gamma, V)$ read
\be
O =  A\(O_1\otimes O_2\otimes O_{3}\)A,
\ee
with
\be
\begin{split}
&O_1 = \sigma_1(\vect{a})\ket{L}_{11}\bra{L},
\qquad
O_2 = \sigma_2(\vect{b})\ket{C}_{22}\bra{C},
 \qquad
 O_3 =\ket{\gamma}_{34\,\,34}\bra{\gamma}.
\end{split}
\ee
Under the state $\ket{\Psi(\alpha,\beta,\gamma)}\in S(\Gamma,V)$, the observable $O$ has the expectation value  which is factorized as
\be
\bra{\Psi(\alpha,\beta,\gamma)}O\ket{\Psi(\alpha,\beta,\gamma)}={}_1\bra{\alpha}\sigma_1(\vect{a})\ket{\alpha}_{1\;2}\bra{\beta}\sigma_2(\vect{b})\ket{\beta}_{2\;34}\bra{\gamma}O_3\ket{\gamma}_{34}.
\ee

On the other hand, the factorizability does not hold for states which do not belong to $S(\Gamma,V)$. For example, we have the state
\be
\ket{\Phi} = A\(\ket{\psi^+(\alpha,\alpha^\prime)}\ot\ket{L}_1\ket{C}_2\ot\ket{+}_3\ket{R}_3\ket{-}_4\ket{R}_4\)
\ee
which does not belong to $S(\Gamma,V)$ due to the entangled spin part 
$\ket{\psi^+(\alpha,\alpha^\prime)}$ in (\ref{alplus}). The expectation value of the state is given by
\be
\begin{split}
\bra{\Phi}O\ket{\Phi}=\frac{N(\alpha,\alpha^\prime)^2}{2}&\Bigl({_1}\bra{\alpha}\sigma_1(\vect{a}) \ket{\alpha}_1\, {_2}\bra{\alpha^\prime} \sigma_2(\vect{b})\ket{\alpha^\prime}_2
+{_1}\bra{\alpha^\prime} \sigma_1(\vect{a}) \ket{\alpha^\prime}_1\,{_2}\bra{\alpha}\sigma_2(\vect{b}) \ket{\alpha}_2\\
&+{_1}\bra{\alpha}\sigma_1(\vect{a}) \ket{\alpha^\prime}_1\, {_2}\bra{\alpha^\prime} \sigma_2(\vect{b})\ket{\alpha}_2
+{_1}\bra{\alpha^\prime} \sigma_1(\vect{a}) \ket{\alpha}_1\,{_2}\bra{\alpha}\sigma_2(\vect{b}) \ket{\alpha^\prime}_2\Bigr)\cdot
{}_{34}\bra{\gamma}O_3\ket{\gamma}_{34},
\end{split}
\ee
which is not factorized with respect to the partition $\Gamma$ in (\ref{quatro}).

\section{Conclusions and Discussions}
 \setcounter{equation}{0}  

In this paper, we have presented a general and coherent framework to examine the separability of states for $N$-particle fermionic systems under arbitrary partitions. The system we have considered consists of $N$ constituent systems of fermions each possessing the $d$-dimensional Hilbert space $\mathbb{C}^d$, and the decomposition of the system into subsystems is specified by an arbitrary partition $\Gamma$ which provides a grouping of the constituents in the system. With the help of an orthogonal structure $V$ defined in the constituent space $\mathbb{C}^d$, we have found in the entire Hilbert space ${\cal H}_A^{\rm tot}$ a set $S(\Gamma, V)$ of states for which the inner products are factorizable according to the partition $\Gamma$. Moreover, we have seen that there exists a class of observables ${\cal C}(\Gamma, V)$ on ${\cal H}_A^{\rm tot}$ whose expectation values (and weak values) are factorizable for all states in $S(\Gamma, V)$ (Theorem 1 and Corollary 1).   
In the $N=2$ example mentioned in Sect.~II, we introduced the additional 
(spatial) degrees of freedom to realize an orthogonal structure $V$ in order to recover the behavior of distinguishable particles for fermionic states.  
However, our argument shows that the same goal is achieved without invoking such additional 
degrees of freedom, if the constituent system has a sufficiently large dimension $d$ to sustain fermionic states within the internal degrees of freedom alone, {\it i.e.}, if $d \ge N$.  

Our separability criterion (Definition 1) for arbitrary $s$ partitions,
given by a set of projection operators, is an extension (with a slight
modification in form) of the criterion proposed earlier in \cite{GMW02} for
the case $s=2$ and $s=N$ in the context of \lq complete sets of properties\rq.
Our Theorem 2 shows that our criterion is fulfilled by a state if and only if it
belongs to $S(\Gamma, V)$ for some $V$ under a given $\Gamma$, that is,
$S(\Gamma, V)$ furnishes
the entire set of states which are separable with respect to the partition
$\Gamma$ and the orthogonal structure $V$, generalizing the result obtained
in \cite{GMW02} for the fermionic case.

A salient feature of our criterion, or of multi-partite separability in
general,  is the explicit dependence on the orthogonal structure $V$.  Most
notably, it implies that a state which is separable under $V$ may no longer
be separable under different $V$ even with the same partition $\Gamma$.
In a sense this is natural because, on physical grounds, the choice of $V$
corresponds to the available experimental setups prepared for measurement,
where the orthogonal structure $V$ ensures that the measurement yields
unambiguous outputs concerning the separation of the subsystems.
The class of observables ${\cal C}(\Gamma, V)$ for which the state possesses
factorizable expectation values will then be determined by the setups used.
Indeed, the example of spin measurement for remotely separated fermions
given in Sect.~V shows that localized measurement setups have the class
${\cal C}(\Gamma, V)$ of observables for which no correlations arise with
the set $S(\Gamma, V)$ of anti-symmetrized product states
in line with the cluster decomposition property.  However, the same class of
observables admit
inseparable states which are entangled in the spin part and the spatial
part, separately.  
In addition, we have mentioned the non-uniqueness in the assignment of
$S(\Gamma, V)$ to a separable state, as well as
the non-transitivity of the separability, as important consequences of  the
$V$-dependence of our separability criterion.
These considerations on the implications of the orthogonal structure $V$
suggest that it might be actually more appropriate to ask the separability
of the system based on
the experimental setups qualified as measurement for separation, not on the
states as we do. 

In closing we remark that our framework can in principle be extended to the
bosonic (totally symmetric) states by use of the symmetrizer instead of the
anti-symmetrizer, coupled with some extra consideration pertinent to the
bosonic systems.   It should also be important to study how to quantify the
amount of inseparability, namely, to find suitable entanglement measures for
$N$-particle indistinguishable particles relevant to the separability under
arbitrary partitions.   The family of measures proposed recently by us in
\cite{IST09} which are specifically designed to examine the arbitrary
separability (for distinguishable particles) may be convenient for that
purpose.

\section*{ACKNOWLEDGEMENTS}
The authors thank Prof. Akira Shimizu for calling our attentions to the topic.  
T. I. is supported by \lq Open Research Center\rq~Project for Private Universities; matching fund subsidy, T. S. is supported by Global COE Program \lq the Physical Sciences Frontier\rq, and I. T. is supported by the Grant-in-Aid for Scientific Research (C), No.~20540391-H21, all of MEXT, Japan.


\end{document}